\documentclass[a4paper,USenglish,cleveref, autoref, thm-restate]{lipics-v2021}

\pdfoutput=1

\hideLIPIcs
\nolinenumbers

\usepackage{dsfont}
\usepackage{mathtools}
\usepackage[ruled, linesnumbered]{algorithm2e}
\usepackage[table, dvipsnames]{xcolor}
\definecolor{myYellow}{rgb}{1 1 0.878}

\newcommand*{\N}{\ensuremath{\mathds{N}}}

\newcommand*{\R}{\ensuremath{\mathds{R}}}

\newcommand*{\T}{\ensuremath{\mathcal{T}}}
\DeclareMathOperator{\SW}{SW}
\DeclareMathOperator{\OPT}{OPT}
\DeclareMathOperator{\POA}{PoA}
\DeclareMathOperator{\POS}{PoS}
\newcommand*{\x}[1]{x^{(#1)}}

\usepackage{mathtools}

\allowdisplaybreaks

\title{Social Distancing Network Creation} 

\author{Tobias Friedrich}{Hasso Plattner Institute, University of Potsdam, Germany}{friedrich@hpi.de}{}{}

\author{Hans Gawendowicz}{Hasso Plattner Institute, University of Potsdam, Germany}{hans.gawendowicz@hpi.de}{}{}

\author{Pascal Lenzner}{Hasso Plattner Institute, University of Potsdam, Germany}{pascal.lenzner@hpi.de}{}{}

\author{Anna Melnichenko}{Hasso Plattner Institute, University of Potsdam, Germany}{anna.melnichenko@hpi.de}{}{}

\authorrunning{T. Friedrich and H. Gawendowicz and P. Lenzner and A. Melnichenko} 

\Copyright{Tobias Friedrich and Hans Gawendowicz and Pascal Lenzner and Anna Melnichenko} 



\category{} 

\relatedversion{} 

\funding{This work was supported by the DFG project GEONET under grant DFG 442003138.}

\begin{document}

\maketitle

\begin{abstract}
During a pandemic people have to find a trade-off between meeting others and staying safely at home. While meeting others is pleasant, it also increases the risk of infection. We consider this dilemma by introducing a game-theoretic network creation model in which selfish agents can form bilateral connections. They benefit from network neighbors, but at the same time, they want to maximize their distance to all other agents. This models the inherent conflict that social distancing rules impose on the behavior of selfish agents in a social network. Besides addressing this familiar issue, our model can be seen as the inverse to the well-studied Network Creation Game by Fabrikant et al.~[PODC 2003] where agents aim at being as central as possible in the created network. Thus, our work is in-line with studies that compare minimization problems with their maximization versions.

We look at two variants of network creation governed by social distancing. In the first variant, there are no restrictions on the connections being formed. We characterize optimal and equilibrium networks, and we derive
asymptotically tight bounds on the Price of Anarchy and Price of Stability. The second variant is the model's generalization that allows restrictions on the connections that can be formed. 
As our main result, we prove that Swap-Maximal Routing-Cost Spanning Trees, an efficiently computable weaker variant of Maximum Routing-Cost Spanning Trees, actually resemble equilibria for a significant range of the parameter space. Moreover, we give almost tight bounds on the Price of Anarchy and Price of Stability. These results imply that, compared the well-studied inverse models, under social distancing the agents' selfish behavior has a significantly stronger impact on the quality of the equilibria, i.e., allowing socially much worse stable states.
\end{abstract}

\section{Introduction}
\label{sec:introduction}
Network Design is a core topic in Theoretical Computer Science and Operations Research. Many classical combinatorial optimization problems, inspired by real world applications, have been formulated and analyzed, such as the \textsc{Minimum Spanning Tree} problem \cite{graham1985}, the \textsc{Network Design} problem \cite{johnson1978,MW84} and finding geometric spanners \cite{bose2013,narasimhan2007geometric}. Typically, 
a network having certain properties must be found by a centralized algorithm.
However, in many settings, the desired network is not created by a central authority but by individually acting agents, e.g., people or institutions, controlling 
a local part of the network.
Prominent examples are the Internet, road networks, and, most relevant for our work, social networks.

Especially in settings with little coordination, these individual agents tend to selfishly optimize their own utility without taking the impact of their actions on the efficiency of the whole network into account.
To better understand the dynamics arising in these decentralized settings and the network structures resulting from them, many influential game-theoretic network formation models have been introduced in the last decades~\cite{jackson1996,BG00,fabrikant2003,anshelevich2008,anshelevich2008_2}. 
The main research questions are: Do equilibrium networks, i.e., stable networks where no agent can improve by performing a local change, exist? What properties do these networks have? And how efficient are they compared to centrally computed optimal solutions?

All of the above mentioned influential game-theoretic network formation models assume that the creation of an edge is costly but the agents benefit from having small distances to other agents in the network. 
However, departing from this standard assumption in the field, there are real-world settings that should better be modeled via an inverted utility function: neighbors yield benefit but being close to many agents is costly as it yields an increased risk. One example for this choice are financial networks.  There, financial institutions benefit from working together but suffer from risks arising from one of them failing\footnote{The financial crisis in the late 2000s was mainly driven by contagious network effects of failing banks.}. Another example, that is the main motivation of our work, came up with the current COVID-19 pandemic and is described by the now commonly used term \emph{social distancing}.
It refers to reducing social contacts in order to contain the spread of a contagious virus in the population.
While often mandated by the government, social distancing was performed by many people voluntarily.
One of the main reasons is quite simple:
While reducing social contacts is a restriction of the quality of life, it also reduces the probability of getting infected. Hence, the network of social interactions between people was sparsified by individual strategic decisions.

In this work we introduce a novel game-theoretic network formation model in which selfish agents strategically form a social network under the influence of social distancing. Agents benefit from direct connections to other agents, modeling the positive effects of social contacts on their social life. However, at the same time they want to maximize their distances to all other agents in the network in order to reduce their risk of getting infected via an increased reaction time in case a contagious disease starts spreading in the network. Here we assume that a random network node becomes infected and that it is beneficial to be far away from the source of infection in order to gain valuable time for setting up counter-measures.

The agents in our model act according to an inverted utility function, compared to the famous models by Jackson and Wolinsky~\cite{jackson1996} and Fabrikant et al.~\cite{fabrikant2003}. Thus, to the best of our knowledge, this is one of the rare cases of a game-theoretic model where both minimizing and maximizing the utility function has a natural interpretation. Another similar well-known example is the contrast between the Network Design Game with fair cost sharing by Anshelevich et al.~\cite{anshelevich2008} and the Selfish Routing model by Roughgarden and Tardos~\cite{roughgarden2002bad}. In both models the agents select paths in a given network but in the former sharing an edge is beneficial for the involved agents whereas in the latter edge sharing is detrimental. This difference yields vastly different behavior in terms of the quality of the equilibria. However, this is not obvious, as can also be seen by comparing  
classical minimization and maximization variants of optimization problems, e.g., \textsc{Minimum Spanning Tree} versus \textsc{Maximum Spanning Tree} or \textsc{Shortest Path} versus \textsc{Longest Path}. Sometimes, as with spanning trees, the inverse problems are almost identical, whereas sometimes, as with the path problems, the inverse problems may have completely opposite behavior.
We set out to explore this comparison for the natural inverse counter-part to the well-known Network Creation Game by Fabrikant et al.~\cite{fabrikant2003}.
Along the way, we will uncover a connection to the \textsc{Maximum Routing-Cost Spanning Tree} problem that is inverse to the well-studied \textsc{Minimum Routing-Cost Spanning Tree} problem~\cite{hu1974optimum}\footnote{This problem is also known as the \textsc{Optimum Communication Spanning Tree} problem.}. 

\subsection{Model and Notation}
\label{sec:modelNotation}
Before we start with the model definition, we introduce some notation regarding networks. A \emph{network} is a tuple $G\coloneqq (V,E)$ where $V$ is the set of \emph{nodes} and $E$ is the set of \emph{edges}. An \emph{edge} is represented by a set containing both incident nodes. If we do not give the tuple defining $G$ explicitly, we denote the set of nodes of $G$ as $V_G$ and the set edges of $G$ as $E_G$. We only consider unweighted undirected networks. For addition and removal of a single edge $e$, we write $G+e\coloneqq (V,E\cup\{e\})$ and $G-e\coloneqq (V,E\setminus\{e\})$. A network $G'$ with $V_{G'}\subseteq V$ and $E_{G'}\subseteq E$ is called a \emph{subnetwork} of $G$ and denoted as $G'\le G$. 
If $G'$ is connected and $V_{G'}=V$, $G'$ is a \emph{spanning subnetwork} of $G$. Let $n\in\N$ denote the number of nodes. The set of all connected networks containing exactly $n$ nodes will be referred to as $\mathcal{G}_n$.

For two nodes $v,x\in V$, we define $d_G(v,x)$ as the \emph{distance} between $v$ and $x$ in network~$G$, that is, the number of edges on a shortest path from $v$ to $x$ in $G$. For convenience, we extend the definition of $d_G$ to sets of nodes: Let $v\in V$ be a node and $M,N\subseteq V$ be sets of nodes. Then $d_G(v,M)\coloneqq \sum_{x\in M}d_G(v,x)$ and $d_G(M,N)\coloneqq \sum_{x\in M, y\in N}d_G(x,y)$. We call the special case $d_G(v,V)$ the \emph{distances from/for} $v$ and $d_G(V,V)$ the \emph{total/summed distances} or \emph{routing costs} of $G$. The \emph{degree} of $v$ in the network $G$ is the number of edges that are incident to $v$ and is denoted as $\deg_G(v)$. We call a tree which is a spanning subnetwork of $G$ a \emph{spanning tree} of $G$. A spanning tree of $G$ with routing costs at least as high as the routing costs of any other spanning tree of $G$ will be called a \emph{Maximum Routing-Cost Spanning Tree} (MRCST). A spanning tree of $G$ with routing costs that cannot be increased by swapping one edge is a \emph{Swap-Maximal Routing-Cost Spanning Tree} (SMRCST).

Now, we can define the game-theoretic model. Let $H=(V,E)$ be a connected network. We call $H$ the \emph{host network} and its nodes \emph{agents}. A \emph{state} of the game $G\le H$ is a spanning subnetwork of $H$. 
We only consider connected networks as host networks and states.

Each agent $v\in V$ selfishly tries to maximize its utility in state $G$ given by
\begin{equation*}
	u_v(G)\coloneqq \alpha\deg_G(v)+d_G(v,V)
\end{equation*}
where $\alpha\in\R_{>0}$ is a global parameter. We will call $\alpha\deg_G(v)$ the \emph{edge utility} and $d_G(v,V)$ the \emph{distance utility} of $v$. Note that $\alpha$ is a parameter of the game, i.e., equal for all agents, that allows to adjust the agents' trade-offs between edge utility and distance utility.  Here $\alpha$ is the benefit of a single edge, i.e., the benefit for each direct neighbor in the network.

For measuring the efficiency of the network $G$, we use the \emph{social welfare} defined as $\SW(G)\coloneqq \sum_{v\in V_G}u_v(G)=2\alpha|E_G|+d_G(V,V)$. This quantifies the well-being of the society of all agents. We call a network maximizing the social welfare for the host network $H$ a \emph{social optimum} and denote it as $\OPT_H$.

Agents are allowed to form connections bilaterally. More specifically, each agent can unilaterally remove any incident edge if it does not disconnect the network, and two agents together can form an edge between them if it is contained in the host network. If removing an edge strictly increases the utility of one of its incident nodes or adding an edge strictly increases the utility of both incident nodes, we call this an \emph{improving move}. A network without improving moves is referred to as \emph{pairwise stable} \cite{jackson1996} or \emph{stable} for short\footnote{As shown by Corbo and Parkes~\cite{CP05} for bilateral Network Creation Games, pairwise stability is equivalent to pairwise Nash stability, which is a refinement of the Nash equilibrium: it must be stable against unilateral deviations and it must be stable against joint strategy changes by coalitions of agents of size two. The strategy space of any agent $i\in V$ is the power set of $V\setminus{i}$. An edge $\{u,v\}$ is formed if and only if $v$ is in agent $u$'s strategy and $u$ is in agent $v$'s strategy.}. If there are no improving edge additions or removals we call the network \emph{stable against edge addition} and \emph{stable against edge removal}, respectively.

For a host network $H$, we define $\mathcal{S}(H)$ as the set of all pairwise stable states. For measuring the efficiency lost by letting agents form the network selfishly, we use the \emph{Price of Anarchy}~(PoA)~\cite{KP99} and \emph{Price of Stability}~(PoS)~\cite{anshelevich2008} defined as
$$
	\POA_n\coloneqq \max_{H\in\mathcal{G}_n}\max_{G\in\mathcal{S}(H)}\frac{\SW(\OPT_H)}{\SW(G)}
	\textnormal{~~~~~and~~~~~}
	\POS_n\coloneqq \max_{H\in\mathcal{G}_n}\min_{G\in\mathcal{S}(H)}\frac{\SW(\OPT_H)}{\SW(G)}.
$$

We will call this model \emph{Social Distancing Network Creation Game} (SDNCG). In~\Cref{sec:complete} we will restrict the host networks to complete networks $K_n$. We will call this restricted variant \emph{complete Social Distancing Network Creation Game} ($K$-SDNCG).

\subsection{Related Work}
\label{sec:relatedWork}
Variants of game-theoretic network formation models have been studied extensively for decades and we refer to Jackson~\cite{jackson2010social} for an overview. 

Closest to our work is the literature on the Network Creation Game~(NCG) by Fabrikant et al.~\cite{fabrikant2003}. This influential model can be seen as the unilateral inverted variant of the $K$-SDNCG. There, an agent can buy any incident edge without the consent of the other endpoint for the price of $\alpha > 0$. Each agent aims at minimizing its cost, which is defined as the sum of $\alpha$ times the number of bought edges and the sum of hop-distances to all agents. The authors of~\cite{fabrikant2003} show that Nash equilibria always exist, i.e., complete networks are stable for $\alpha \leq 2$ and stars are stable for $\alpha \geq 2$. However, besides these generic examples finding Nash equilibria is challenging since the NCG and many of its variants do not belong to the class of potential games~\cite{Lenzner11,KL13}. Besides finding equilibria, also computing a best possible strategy is challenging, since this problem was shown to be NP-hard in~\cite{fabrikant2003}. However, such strategies can be efficiently approximated with greedy strategy changes~\cite{Lenzner12}.  
Regarding the quality of equilibrium states the authors of~\cite{fabrikant2003} show that the PoA is in $\mathcal{O}(\sqrt{\alpha})$, that the PoA for tree Nash equilibria is constant, and that the PoS is at most $\frac43$. Later, a series of papers~\cite{Al14,De07,MS10,MMM13,AM17,BiloL20,AM19} improved the PoA bounds, with the best general upper bound of $2^{\mathcal{O}\sqrt{\log n}}$ by Demaine et al.~\cite{De07}. The latter also proved that the PoA is constant for $\alpha \in \mathcal{O}(n^{1-\varepsilon})$ for any fixed $\varepsilon > \frac{1}{\log n}$. For large $\alpha$, it was shown by Bilò and Lenzner~\cite{BiloL20} that for $\alpha > 4n-13$ all Nash equilibria must be trees and this bound was recently improved by Dippel and Vetta~\cite{Dippel21} to $\alpha > 3n-3$. This implies a constant PoA for $\alpha > 3n-3$. Finally, Álvarez and Messegué~\cite{AM19} established a constant PoA for $\alpha > n(1+\varepsilon)$, for any $\varepsilon > 0$.

The NCG was generalized by Demaine et al.~\cite{demaine2009} by introducing a host network that specifies which edges can be bought. They show that the PoA deteriorates by providing a lower bound of $\Omega(\min\{\alpha/n, n^2/\alpha\})$ and an upper bound of $\mathcal{O}(\sqrt{\alpha})$, for $\alpha < n$, and $\mathcal{O(}\min\{\sqrt{n},n^2/\alpha\})$, for $\alpha \geq n$. Interestingly, no results on the existence of equilibria are known. Recently, a further generalization that allows weighted host networks was proposed by Bilò et al.~\cite{bilo2019}. This variant has a tight PoA of $(\alpha+2)/2$ for metric weights. Later a tight bound of $\Theta(\alpha)$ was shown for arbitrary weights~\cite{friedemann2021}.
Also a bilateral variant of the NCG was studied by Corbo and Parkes~\cite{CP05}. There, similar to our model, edges can only be established by bilateral consent of the involved nodes and both nodes have to pay $\alpha$. The authors of~\cite{CP05} prove existence of pairwise stable networks, i.e., complete networks are stable for $\alpha \leq 1$ and stars are stable for $\alpha  \geq 1$, 
they give a tight PoA bound of $\Theta(\min\{\sqrt{\alpha},n/\sqrt{\alpha}\})$, and they show that the PoS is $1$. To the best of our knowledge, the bilateral variant with a given host network has not yet been studied. Recently, also a bilateral variant modeling the formation of social networks was introduced~\cite{BFLLM21}.   

The idea of a game-theoretic model of network formation in a context of spreading risk is not new. Goyal et al.~\cite{Goyal16} study a setting where a node is attacked and this attack spreads to all vulnerable neighbors. Agents strategically create edges and immunize themselves to maximize their connected component post attack. For this model, also the efficient computation of best strategies~\cite{Friedrich17} and a variant with probabilistic spread~\cite{Chen19} was studied.
Moreover, there has been much research in the context of financial contagion, where agents benefit from collaborating, but also suffer from the risk of cascading failure arising with the collaboration~\cite{allen2000,haldane2011,caballero2009,acemoglu2015}. In particular, Blume et al.~\cite{blume2011} developed an elegant model where nodes form a network and then some randomly chosen nodes fail and this failure then spreads with some probability via the edges. The utility is a linear combination of the node degree and the risk of failing in the second phase. The virtue of this model is that  utilities are based on a random process that realistically models the spread of a contagious infection. However, the major downside of this model is that the computation of the random process is \#P-complete. Thus, this model does not yield a realistic prediction of real-world behavior.

While analyzing our model for general host networks, we consider Maximum Routing-Cost Spanning Trees. Routing costs have been studied much in mathematics, mostly under the name of the \emph{Wiener index} \cite{wiener1947}. Trees were of special interest and there has been much research on the Wiener index of trees with different properties.
But although spanning trees minimizing the Wiener index were studied extensively, the concept of spanning trees maximizing the Wiener index received little attention \cite{dobrynin2001,xu2014}. However, it was shown that finding or even approximating a tree maximizing the Wiener index is NP-hard \cite{CAMERINI83, GALBIATI97}.

\subsection{Our Contribution}
We introduce the Social Distancing Network Creation Game (SDNCG), a game-theoretic model in which selfish agents try to maximize their utility by strategically connecting to other agents and thereby creating a network. Each agent values direct connections to other agents but at the same time wants to maximize the distances to all other agents in order to lower their exposure and increase their reaction time to risks appearing in the network. In contrast to the similar model by Blume et al.~\cite{blume2011}, our model, while not modeling a perfectly realistic spread of the infection, has the advantage of an efficiently computable utility function. By using the distance to the other agents as part of the utility, it also accounts for reaction time: If an infection breaks out far away, an agent has more time to prepare or react to it.
Another virtue of our model is that it is the inverse to the well-known Network Creation Game~\cite{fabrikant2003} and its bilateral variant~\cite{CP05}. Hence, we can study and compare the game-theoretic properties of the inverted models. To the best of our knowledge, this is one of the rare cases where both the minimization and the maximization of a utility function have a natural interpretation.  

Our results and the comparison with the inverted models are summarized in Table~\ref{table:results}.
\begin{table}[h]
	\centering\setlength{\tabcolsep}{3pt}%
	\renewcommand{\arraystretch}{1.3}%
	\resizebox{\textwidth}{!}{\begin{tabular}{c|c|c|c|c}
		& Optimum & Equilibria & PoA & PoS
		\\ \hline
		NCG~\cite{fabrikant2003} & 
			\begin{tabular}{@{}l@{}}$\alpha\leq 2\colon K_n$~\cite{fabrikant2003}\\$\alpha\ge 2\colon S_n$~\cite{fabrikant2003}\end{tabular} &
			\begin{tabular}{@{}l@{}}$\alpha\leq 1\colon K_n$~\cite{fabrikant2003}\\$\alpha\geq 1\colon S_n$~\cite{fabrikant2003}\end{tabular} &
			\begin{tabular}{@{}l@{}}$2^{\mathcal{O}\left(\sqrt{\log n}\right)}$~\cite{De07}\\$\alpha\in\mathcal{O}\left(n^{1-\varepsilon}\right)\colon\Theta(1)$~\cite{De07}\\$\alpha>n(1+\varepsilon)\colon \Theta(1)$~\cite{AM19}\end{tabular} &
			\begin{tabular}{@{}l@{}}$\alpha\le 1\colon 1$~\cite{fabrikant2003}\\$1{<}\alpha{<}2\colon\leq \frac43$~\cite{fabrikant2003}\\$\alpha\ge 2\colon 1$~\cite{fabrikant2003}\end{tabular}
		\\ \hline
		BNCG~\cite{CP05} &
			\begin{tabular}{@{}l@{}}$\alpha<1\colon K_n$~\cite{CP05}\\$\alpha>1\colon S_n$~\cite{CP05}\end{tabular} &
			\begin{tabular}{@{}l@{}}$\alpha<1\colon K_n$~\cite{CP05}\\$\alpha>1\colon S_n,\dots$~\cite{CP05}\end{tabular} &
			\begin{tabular}{@{}l@{}}$\Theta\left(\min\left\{\sqrt{\alpha},\frac{n}{\sqrt{\alpha}}\right\}\right)$~\cite{CP05}\\$\alpha<1\colon 1$~\cite{CP05}\end{tabular} &
			1~\cite{CP05}
		\\ \hline\rowcolor{myYellow}
		$K$-SDNCG &
			\begin{tabular}{@{}l@{}}$\alpha<\frac{n}{3}\colon P_n$ [T. \ref{thm:com:opt}]\\$\alpha>\frac{n}{3}\colon K_n$ [T. \ref{thm:com:opt}]\end{tabular} &
			\begin{tabular}{@{}l@{}}$\alpha\le 1\colon$trees [T. \ref{thm:com:stable}]\\$1\le\alpha\le\frac{n}{2}\colon$\\
			$\qquad P_n, K_n, \dots$ [T. \ref{thm:com:stable}]\\$\alpha\ge\frac{n}{2}\colon K_n$ [T. \ref{thm:com:stable}]\end{tabular} &
			\begin{tabular}{@{}l@{}}$\mathcal{O}(n)$ [T. \ref{thm:com:poa}]\\$\alpha\le\sqrt{n}\colon\Theta(n)$ [T. \ref{thm:com:poa}]\\$\alpha{\le}\frac{n}{6}\!-3\colon\Omega\left(\frac{n}{\log n}\right)$ [T. \ref{thm:com:poa}]\\$\alpha{\le}\left\lfloor\frac{n}{2}\right\rfloor\!-2\colon\Omega\left(\sqrt{n}\right)$ [T. \ref{thm:com:poa}]\\$\alpha\ge\frac{n}{2}\colon 1$ [T. \ref{thm:com:poa}]\end{tabular} &
			1 [T. \ref{complete:PoS}]
		\\ \hline
		$H$-NCG~\cite{demaine2009} &
			open &
			open &
			\begin{tabular}{@{}l@{}}$\alpha< n\colon\mathcal{O}\left(\sqrt{\alpha}\right)$~\cite{demaine2009}\\$\alpha{\ge} n\colon\!\!\min\left\{\!\mathcal{O}\!\left(\!\sqrt{n},\frac{n^2}{\alpha}\!\right)\!\right\}\!$~\cite{demaine2009}\\$\Omega\left(\min\left\{\frac{\alpha}{n},\frac{n^2}{\alpha}\right\}\right)$~\cite{demaine2009}\end{tabular} &
			open 
		\\ \hline\rowcolor{myYellow}
		SDNCG & 
			\begin{tabular}{@{}l@{}}$\alpha\le 1\colon\!$MRCST [T. \ref{thm:gen:opt}]\\$\alpha>N_3\colon H$ [T. \ref{thm:gen:opt}]\end{tabular} &
			\begin{tabular}{@{}l@{}}$\alpha\le 1\colon$trees [T. \ref{thm:gen:stable}]\\$1\le\alpha\le\frac{n}{3}\colon$\\$\qquad$SMRCST [T. \ref{general:stable:huge_proof}]\\$\alpha\ge N_2\colon H$ [T. \ref{thm:gen:stable}]\end{tabular} &
			\begin{tabular}{@{}l@{}}$\mathcal{O}(n)$ [C. \ref{thm:gen:poa}]\\$\alpha{\le} n\colon\Theta(n)$ [T.\ref{thm:gen:poa}]\\$\alpha\le N_2\colon\Omega\left(\frac{n^2}{\alpha}\right)$ [T. \ref{thm:gen:poa}]\\$N_2{<}\alpha{\le} N_3\colon\Theta(1)$ [T. \ref{thm:gen:poa}]\\$\alpha\ge N_3\colon 1$ [T. \ref{thm:gen:poa}]\end{tabular} &
			\begin{tabular}{@{}l@{}}$\alpha\le 1\colon 1$ [T. \ref{thm:gen:pos}]\\$\alpha<\frac{n}{3}\colon\mathcal{O}\left(\sqrt{n}\right)$ [T. \ref{thm:gen:pos}]\\$N_2{<}\alpha{\le} N_3\colon\!\Theta(1)$ [T. \ref{thm:gen:pos}]\\$\alpha\ge N_3\colon 1$ [T. \ref{thm:gen:pos}]\end{tabular}
	\end{tabular}}
	\caption{An overview of our results (yellow) and a comparison with the results for the inverted models (white). BNCG abbreviates the bilateral NCG by Corbo and Parkes~\cite{CP05} whereas $H$-NCG denotes the NCG on a host network by Demaine et al.~\cite{demaine2009}.  $N_2\coloneqq\frac{(n-1)^2}{4}$, $N_3\coloneqq\frac{(n-2)n(n+2)}{24}$, $H$ denotes the host network, $P_n,K_n,S_n$ are the path, clique, and star networks on $n$ nodes, respectively.}
	\label{table:results}
\end{table}
We analyze two variants of the SDNCG. For the $K$-SDNCG, where we assume a complete host network, we 
characterize optimal and several stable networks and show that the PoS is $1$.
We provide an improving response cycle, which implies that equilibrium existence for the ($K$-)SDNCG cannot be derived from potential function arguments. Finally, derive several bounds for the PoA which are tight for $\alpha\ge\frac{n}{2}$, asymptotically tight for $\alpha\le \sqrt{n}$, and asymptotically tight up to a $\log$-factor for $\alpha\le\frac{n}{6}-3$.

For the SDNCG on arbitrary host networks 
we utilize Maximum Routing-Cost Spanning Trees for characterizing optimal networks for $\alpha \leq 1$. As our main result, we show that their locally optimal variant, the Swap-Maximal Routing-Cost Spanning Trees, and hence also Maximum Routing-Cost Spanning Trees, are pairwise stable for $\alpha\le\frac{n}{3}$. 
We prove that computing the MRCST is NP-hard, while the SMRCST can be constructed efficiently. Thus, for the significant range of $1\leq \alpha \leq \frac {n}{3}$, we not only have guaranteed equilibrium existence on any host graph, but we can compute stable states efficiently. This is in stark contrast to what is known for the inverse model studied by Demaine et al.~\cite{demaine2009}.
Additionally, 
we approximate optimal networks and we derive several (tight) bounds on the PoA and the PoS.

Compared with the NCG~\cite{fabrikant2003} and the bilateral NCG~\cite{CP05}, we find that the results for the $K$-SDNCG regarding optimal and stable networks are analogous but reversed, with the spanning path taking over the role of the spanning star. Moreover, our PoA results for both the $K$-SDNCG and the SDNCG show that our maximization variant has a significantly worse PoA that is linear or almost linear in $n$, compared to the PoA upper bounds of $o(n^\varepsilon)$ and $\mathcal{O}(\sqrt{\alpha},n/\sqrt{\alpha})$ for the NCG and the bilateral NCG, respectively. As main take away from our paper, this implies that under social distancing the agents' selfish behavior has significantly more impact on the quality of the equilibria. This calls for strong coordination mechanisms governing the network formation to avoid detrimental stable states.


\section{Complete Host Networks}
\label{sec:complete}
We analyze the properties of the $K$-SDNCG, i.e., the SDNCG on complete host networks. 
First, we characterize optimal networks and give some examples for stable networks, dependent on the relation between $n$ and $\alpha$. After that, we show several bounds on the PoA and PoS.

\subsection{Stable and Optimal Networks}
\label{sec:complete:stableOptimal}
Intuitively, for small $\alpha$ the distance utility dominates the social welfare. Hence, the path should be the optimum since it maximizes the total distances. 
For large $\alpha$, the edge utility dominates, which leads to the clique being optimal since it maximizes the number of edges. Now we show that this intuition is indeed true. Moreover, the optimal construction is unique.

\begin{restatable}{thm}{socialOptimum}
	\label{thm:com:opt}
	For $\alpha<\frac{n}{3}$, the unique social optimum is the path. For $\alpha>\frac{n}{3}$, the unique optimum is the clique. For $\alpha=\frac{n}{3}$, the clique and the path are the only social optima.
\end{restatable}
\begin{proof}
	Šoltés and Ľubomír~\cite{soltes1991} showed that for a fixed number of nodes and edges, the network maximizing the summed distances is unique and contains a clique and a path with at least two edges between one endpoint of the path and the clique. We call this a \emph{PathClique}. (Note that the clique can be empty, resulting in just a path) For a visualization, we refer to \Cref{app:fig:path_clique}.
	\begin{figure}[h]
		\centering
		\includegraphics[width=0.6\textwidth]{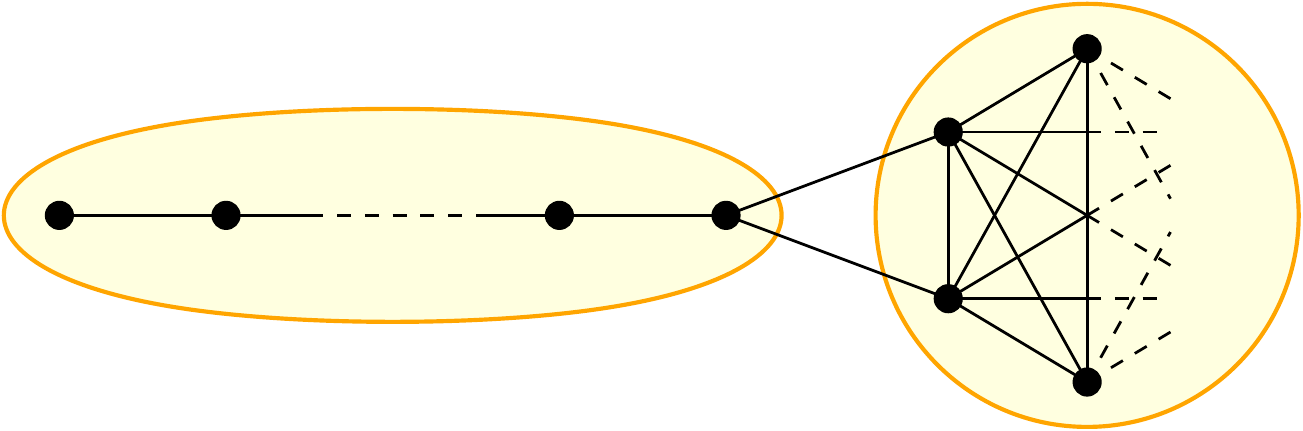}
		\caption{This figure shows a PathClique. It consists of a path (left) and a clique (right), which are connected by at least two edges between one endpoint of the path and some nodes of the clique.}
		\label{app:fig:path_clique}
	\end{figure}
	Note, that the social optimum has to be such a network, since for every other network, there is a PathClique with the same number of edges but larger summed distances and therefore a larger social welfare.
	
	Let $G$ be a PathClique with $n$ nodes having a clique containing $k$ nodes. Then the corresponding path contains $n-k$ nodes. First, we show that, unless $G$ is a path or a clique, we get a socially better network by adding or removing edges.
	
	Let $G$ be neither a clique nor a path and let $v$ be the endpoint of the path that is connected to the clique. Observe that removing an edge between $v$ and the clique results in a PathClique. (This is still true if there are only two edges connecting $v$ to the clique: Removing one of these edges makes the remaining neighbor of $v$ in the clique the new endpoint of the path and reduces the size of the clique by 1.) Therefore, this is the socially best way of removing an edge from $G$. Similarly, the best way of adding an edge to $G$ is adding it between $v$ and the clique unless $v$ is already fully connected to the clique in which case it is best to add an edge between the neighbor of $v$ on the path and the clique.
	We now make a case distinction.
	
	If $v$ is fully connected to the clique, adding an edge decreases distances by $2(n-k-1)$ and deleting an edge increases distances by $2(n-k)$. This means, that $G$ can only be optimal if $\alpha\le 2(n-k-1)$ and $\alpha \ge 2(n-k)$, which is a contradiction.
	
	If $v$ is not fully connected to the clique, adding an edge decreases distances by $2(n-k)$ and deleting an edge increases distances by $2(n-k)$. Thus, $\alpha=2(n-k)$ is necessary for $G$ to be socially optimal. But when $\alpha=2(n-k)$, adding and deleting edges between $v$ and the clique does not change the social welfare. Let $G'$ be the network obtained by fully connecting $v$ to the clique. Then $\SW(G)=\SW(G')$. Additionally, $G'$ cannot be an optimum since it fulfills the conditions of the first case. Therefore, $G$ cannot be socially optimal, too.
	
	So, the social optimum network must be the path $P_n$ or the clique $K_n$. We have
	\begin{align*}
		\SW(P_n)&=2\alpha(n-1)+2\sum_{i=1}^{n-1}i(n-i)=2\alpha(n-1)+\frac{1}{3}(n-1)n(n+1),\\
		\SW(K_n)&=2\alpha\binom{n}{2}+2\binom{n}{2}=n(n-1)(\alpha+1).
	\end{align*}
	For $\alpha=\frac{n}{3}$, we see that
	\begin{align*}
		\SW(P_n)&=\frac{2}{3}n(n-1)+\frac{1}{3}(n-1)n(n+1)=n(n-1)\left(\frac{n}{3}+1\right)=\SW(K_n),
	\end{align*}
	for $\alpha<\frac{n}{3}$ we have $\SW(P_n)>\SW(K_n)$, and $\alpha>\frac{n}{3}$ yields $\SW(K_n)>\SW(P_n)$.
\end{proof}
Next, we have a look at the existence of pairwise stable networks. Similar to the social optimum, for small $\alpha$, agents prefer large distances over many incident edges and therefore should remove as many edges as possible, leading to only trees being stable. Interestingly, the restrictions of pairwise stability lead to all trees being stable for small $\alpha$, even if the distances are very small (like in a star). This is shown by the next theorem.
\begin{restatable}[Stable Networks]{thm}{completeStable}\label{thm:com:stable}\,
	\begin{bracketenumerate}
		\item For $\alpha\le 1$, every tree is pairwise stable. For $\alpha<1$, any pairwise stable network is a tree.\label{thm:com:stable:tree}
		\item For $\alpha\ge 1$, the clique is pairwise stable.\label{thm:com:stable:clique}
		\item For $\alpha\le\frac{n-1}{2}$, the path is pairwise stable.\label{thm:com:stable:path}
		\item For $\alpha>\frac{n}{2}$, the clique is the only pairwise stable network.\label{thm:com:stable:only_clique}
	\end{bracketenumerate}

\end{restatable}
\begin{proof}[Proof of (\ref{thm:com:stable:tree})]
	Let $G$ be a tree. Since removing an edge from $G$ would lead to $G$ being disconnected, we only have to consider adding an edge. This would shorten the distances for both endpoints by at least 1. Since $\alpha\le 1$, this is not an improvement for the agents.
	
	Let $G$ be a network. If $G$ contains a cycle, we can remove an edge without disconnecting the network. By doing this, the distances for both endpoints increase by at least 1. Since $\alpha<1$, this is an improvement for both agents.
\end{proof}
\begin{proof}[Proof of (\ref{thm:com:stable:clique})]
	Let $G$ be a clique. Since adding an edge is not possible, we only look at removing an edge. This would result in a distance increase of 1 for both endpoints. Since $\alpha\ge 1$, this is not an improvement for either of the two incident agents.
\end{proof}
\begin{proof}[Proof of (\ref{thm:com:stable:path})]
	Let $G$ be the path consisting of $n$ nodes $v_1,\dots,v_n$, in that order. Since $G$ is a tree, no edge can be removed without disconnecting the network. Therefore, the only possible move is adding an edge. Let $1\le i<j\le n$ be such that $v_i$ and $v_j$ are not adjacent (i.e., $j-i\ge 2$). If $j\le\frac{n}{2}$, adding the edge $e\coloneqq\{v_i,v_j\}$ shortens distances from $v_i$ to at least $\frac{n}{2}$ nodes. The same holds for $i\ge\frac{n}{2}$ regarding node $v_j$. It is easy to see that the distance decrease is minimal when $j-i=2$. Intuitively, edge $e$ has to be as central as possible. According to this, choosing $i=\left\lfloor\frac{n}{2}\right\rfloor$ and $j=i+2$ yields
	\begin{align*}
		d_G(v_i,V)-d_{G+e}(v_i,V)&=n-j+1=n-i-1\ge\frac{n-2}{2}\text{ and }\\
		d_G(v_j,V)-d_{G+e}(v_j,V)&=i\ge\frac{n-1}{2}.
	\end{align*}
	Thus, if $\alpha\le\frac{n-1}{2}$, adding edge $e$ is not an improving move and $G$ is pairwise stable.
\end{proof}
\begin{proof}[Proof of (\ref{thm:com:stable:only_clique})]
	Let $V$ be a set of $n$ agents and $G=(V,E)$ be a stable network. Let $v\in V$ be a node having minimum total distances, i.e., for all $v'\in V$, we have $d_G(v,V)\le d_G(v',V)$. Let $N_G[v]$ denote the closed neighborhood of $v$.
	
	Now suppose, the network induced by $N_G[v]$ is not a clique. Then there are two neighbors $x,y$ of $v$ with $\{x,y\}\notin E$. We observe that for each node $z\in V$, the distances $d_G(v,z)$ and $d_G(x,z)$ can only differ by 1, since $v$ and $x$ are neighbors. By choice of $v$, there are at least as many nodes that are closer to $v$ than nodes that are closer to $x$. Therefore, there are at most $\frac{n}{2}$ many nodes that are closer to $x$. Adding an edge between $x$ and $y$ can, for node $y$, only shorten distances to nodes which are closer to $x$ than to $v$. Thus, this edge shortens the distances from $y$ by at most $\frac{n}{2}$. The same holds for node $x$. Therefore, for $\alpha>\frac{n}{2}$, this edge would improve the utility of agents $x$ and $y$ and, thus, $G$ would not be stable. This contradicts our assumption. Thus, $N_G[v]$ must induce a clique.
	
	Now let $x$ be a neighbor of $v$. Since $x$ is connected to all neighbors of $v$, we have $d_G(x,V)\le d_G(v,V)$, i.e., also $x$ minimizes pairwise distances. Hence, $N_G[x]$ also induces a clique, leading to $N_G[v]=N_G[x]$. By induction, since $G$ is connected, it must be a clique.
\end{proof}
\Cref{thm:com:stable} implies that socially optimal networks are also stable. In fact, they are stable for a wide range of $\alpha$-values. The clique is stable for $\alpha\ge 1$, meeting the bound below which only trees are stable. Similarly, the path is stable for $\alpha\le\frac{n-1}{2}$, almost meeting the lower bound for only the clique being stable. Additionally, we observe that we only need two networks (path and clique) to provide pairwise stable networks for all possible values of $\alpha$.

For further constructions, we need the following definition. Let $G$ be a network. We call $G'$ a \emph{clique network} of $G$, if it can be obtained by replacing each node of $G$ by a clique of size at least 2 and for each edge of $G$ connect the two corresponding cliques fully bipartite. By using only constant-size cliques, some properties of $G$ (density, length of shortest paths) are preserved while the network is more stable against edge removal.
\begin{restatable}{thm}{cliqueNetwork}
	\label{clique_network}
	Let $G$ be a clique network. For $\alpha\ge 1$, $G$ is stable against edge removal.
\end{restatable}
\begin{proof}
	Removing any edge from $G$ only effects the distances between its endpoints. These distances increase by 1.
\end{proof}

\noindent Finally, we show that stable states may not be found by simply letting agents iteratively play improving moves, i.e., via a sequential process of improving strategy changes. \Cref{fig:best_response_cycle} provides an example of a cyclic sequence of improving moves. 
This also implies that both the $K$-SDNCG and the SDNCG do not belong to the class of potential games~\cite{monderer1996}, i.e., the existence of equilibria cannot be proven via potential function arguments.
\begin{restatable}{thm}{potentialGame}
	The Social Distancing Network Creation Game is not a potential game.
\end{restatable}
\begin{proof}
	This is shown by the existence of improving cycles. See \Cref{fig:best_response_cycle} for an example.
	\begin{figure}[h]
	\centering
	\includegraphics[width=0.6\textwidth]{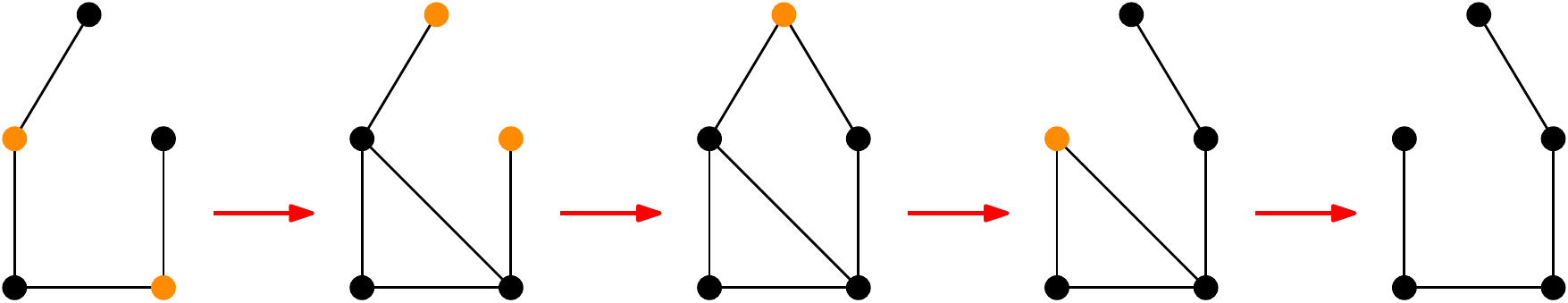}
	\caption{This figure shows a cyclic sequence of improving moves performed by $n=5$ agents for $\alpha=2.5$. In each step, the nodes responsible for the next change are highlighted in orange. Note that the last step is isomorphic to the first step.}
	\label{fig:best_response_cycle}
\end{figure}
\end{proof}

\subsection{Price of Anarchy and Price of Stability}
\label{sec:complete:poaPos}
In this section, we give a series of bounds for the Price of Anarchy and the Price of Stability.
\begin{restatable}[Price of Anarchy]{thm}{completePoa}\label{thm:com:poa}\,
	\begin{bracketenumerate}
		\item The Price of Anarchy is in $\mathcal{O}(n)$.\label{thm:com:poa:upper_bound}
		\item For $\alpha\le 1$, the Price of Anarchy is in $\Theta(n)$.\label{thm:com:poa:alpha_smaller_one}
		\item For $1<\alpha\le\sqrt{n}$, the Price of Anarchy is in $\Theta(n)$.\label{thm:com:poa:alpha_smaller_sqrt}
		\item For $\sqrt{n}\le\alpha\le \frac{n}{6}-3$, the Price of Anarchy is in $\Omega\left(\frac{n}{\log n}\right)$.\label{thm:com:poa:hypercube}
		\item For $\frac{n}{6}-3<\alpha\le\left\lfloor\frac{n}{2}\right\rfloor-2$, the Price of Anarchy is in $\Omega\left(\sqrt{n}\right)$.\label{thm:com:poa:alpha_smaller_n_2}
		\item For $\alpha\ge\frac{n}{2}$, the Price of Anarchy is 1.\label{thm:com:poa:largeAlpha}
	\end{bracketenumerate}

\end{restatable}
\begin{proof}[Proof of (\ref{thm:com:poa:upper_bound})]
	Let $\alpha\in\R_{>0}$. Every connected network has at least $n-1\in\Omega(n)$ and at most $\binom{n}{2}\in\mathcal O(n^2)$ edges. Furthermore, the summed distances are at least $n(n-1)\in\Omega(n^2)$, since every node has at least distance 1 to every other node, and at most $n^3$, since every node has at most distance $n$ to every other node. With this, we get the trivial bound of
	\begin{equation*}
		\POA_n\in\mathcal{O}\left(\frac{\alpha n^2+n^3}{\alpha n+n^2}\right)=\mathcal{O}(n).\qedhere
	\end{equation*}
\end{proof}
\begin{proof}[Proof of (\ref{thm:com:poa:alpha_smaller_one})]
	For $\alpha\le 1$, the social optimum is the path as shown in \Cref{thm:com:opt}. It has social welfare of
	\begin{align*}
		\alpha(n-1)+\frac{1}{3}(n-1)n(n+1).
	\end{align*}
	Because of \Cref{thm:com:stable} and the star being a tree, it is pairwise stable for $\alpha\le 1$. It has social welfare of
	\begin{align*}
		\alpha(n-1)+2(n-1)+2(n-1)(n-2)=\alpha(n-1)+2(n-1)^2.
	\end{align*}
	Together with (\ref{thm:com:poa:upper_bound}), this yields
	\begin{equation*}
		\POA\ge\frac{\alpha(n-1)+\frac{1}{3}(n-1)n(n+1)}{\alpha(n-1)+2(n-1)^2}=\frac{\alpha+\frac{1}{3}n(n+1)}{\alpha+2(n-1)}\in\Theta(n).\qedhere
	\end{equation*}
\end{proof}
\begin{proof}[Proof of (\ref{thm:com:poa:alpha_smaller_sqrt})]
	We construct a star-like network with cliques as leaves in the following way. Let $c\coloneqq\lceil\alpha\rceil+2$. Additionally, let $K_1,\dots,K_d$ be $d\coloneqq\left\lfloor\frac{n-2}{c}\right\rfloor$ cliques containing $c-2$ nodes and $v_1,v_1',v_2,v_2',\dots,v_d,v_d'$ be $2d$ nodes. Let furthermore $M$ be a clique of size $n-cd$. We now define our network $G$ as
	\begin{align*}
		V_G&\coloneqq\bigcup_{i=1}^d V_{K_i} \cup \bigcup_{i=1}^d \{v_i,v_i'\}\cup V_M\\
		E_G&\coloneqq\bigcup_{i=1}^d E_{K_i}\cup E_M \cup\bigcup_{i=1}^d\{\{v_i,v_i'\}\}\\*
		&\hspace{0.6cm}\cup\bigcup_{i=1}^d\bigcup_{v\in K_i}\{\{v,v_i\},\{v,v_i'\}\}\cup\bigcup_{i=1}^d\bigcup_{v\in M}\{\{v,v_i\},\{v,v_i'\}\}.
	\end{align*}
	We essentially connect the outer cliques $K_1,\dots,K_d$ to the center clique $M$ via $d$ 2-cliques and each connection is fully bipartite (see \Cref{fig:star_of_cliques}). Since $n=|V_G|$, $G$ is a network of the desired size.
	\begin{figure}[h]
		\centering
		\includegraphics[width=0.35\textwidth]{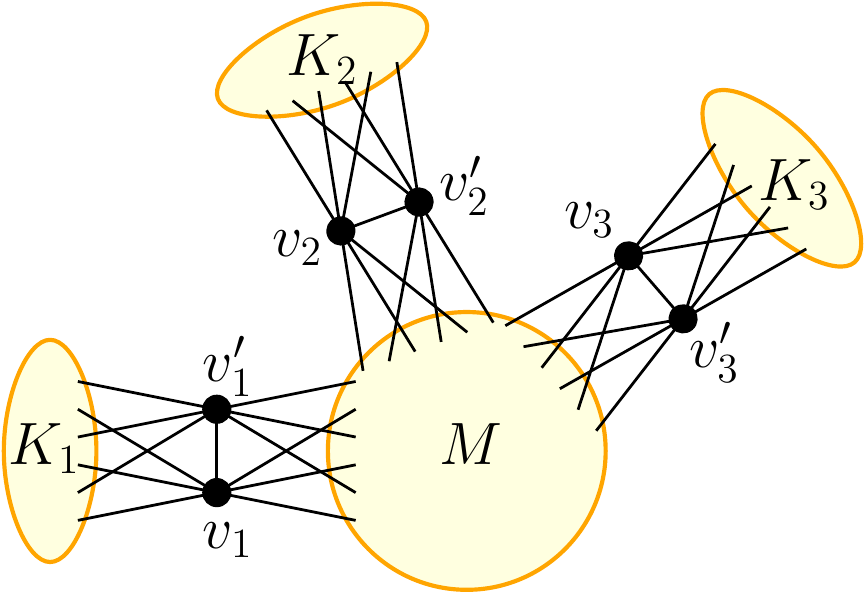}
		\caption{The figure shows a star-like clique network, where the center is formed by a clique $M$ and each ray consists of two nodes $v_i,v_i'$ and a clique $K_i$.}
		\label{fig:star_of_cliques}
	\end{figure}
	
	We now show that $G$ is pairwise stable. We see that $G$ is a clique network. Because of \Cref{clique_network} and $\alpha>1$, $G$ is stable against edge removal. On the other hand, adding an edge shortens distances to at least $|K_i|=c-2\ge\alpha$ nodes which means a distance decrease of at least $\alpha$ for the two incident nodes. This also does not increase their utility. Therefore, $G$ is pairwise stable.
	
	For the center clique $M$, we see that $|V_M|=n-cd=n-c\left\lfloor\frac{n-2}{c}\right\rfloor$ and therefore $2\le |V_M|< n-(n-2-c)=c+2$. With this and $1<\alpha\le\sqrt{n}$, we obtain
	\begin{align*}
		|E_G|&=d\binom{c-2}{2}+\binom{|V_M|}{2}+2d+d(c-2)2+d|V_M|2\\*
		&=\left\lfloor\frac{n-2}{\lceil\alpha\rceil+2}\right\rfloor\left(\frac{\lceil\alpha\rceil(\lceil\alpha\rceil-1)}{2}+2+2\lceil\alpha\rceil+2|V_M|\right)+\frac{|V_M|(|V_M|-1)}{2}\\*
		&\in\Theta(\alpha n)
	\end{align*}
	and $d_G(V_G,V_G)\in\Theta(n^2)$.
	
	For $\alpha<\sqrt{n}$, the socially optimal network is the path. With the previous calculations, we can now bound the Price of Anarchy as
	\begin{align*}
		\POA\ge\frac{2\alpha(n-1)+\Theta(n^3)}{2\alpha\Theta(\alpha n)+\Theta(n^2)}=\frac{\Theta(n^3)}{\Theta(n^2)}\in\Omega(n).
	\end{align*}
	
	From (\ref{thm:com:poa:upper_bound}), we have $\POA\in\mathcal{O}(n)$ and therefore $\POA\in\Theta(n)$.
\end{proof}
\begin{proof}[Proof of (\ref{thm:com:poa:hypercube})]
	Let $d=\lfloor\log n\rfloor-1$. Then, the $d$-dimensional hypercube is represented by $G_H$ with $V_{G_H}=\{0,1\}^d$ and $E_{G_H}=\{\{v,x\}\mid v,x\in V\wedge d_H(v,x)=1\}$ where $d_H(v,x)$ denotes the Hamming Distance between $v$ and $x$.
	Let $G$ be a clique network for $G_H$ with $|V_G|=n$ such that the sizes of the cliques replacing the nodes of $G_H$ differ by at most 1. Observe, that each clique is of size 2 or 3 if $2\cdot 2^d\le n<3\cdot 2^d$ and of size 3 or 4 if $3\cdot 2^d\le n<4\cdot 2^d$.
	By \Cref{clique_network} and since $\alpha\ge 1$, we know that $G$ is stable against edge removal. We now show that adding an edge shortens the total distances for the incident nodes by at least $\frac{n}{6}-3$.

	Let $v,x\in V_G$ such that $e\coloneqq\{v,x\}\notin E_G$ and let $v',x'\in V_{G_H}$ be the nodes corresponding to the cliques that contain $v$ and $x$, respectively. Therefore, $e'\coloneqq\{v',x'\}\notin E_{G_H}$, which implies $d_H(v',x')\ge 2$. By symmetry of the hypercube, we can assume w.l.o.g. that
	$$
		v'=\underbrace{00\dots0}_{d_H(v,x)}\underbrace{0\dots00}_{d-d_H(v,x)} \text{~~~~~~~and~~~~~~~}
		x'=\underbrace{11\dots1}_{d_H(v,x)}\underbrace{0\dots00}_{d-d_H(v,x)}.
	$$
	
	Adding $e'$ to $G_H$ decreases the distances from $v'$ to another node $y'\in V_{G_H}$ if and only if $d_H(v',y')\ge d_H(x,y)+2$. The difference in distance can only come from the first $d_H(v',x')$ bits of the label since the remaining bits are equal for $v'$ and $x'$. Let $\ell$ be the number of the first $d_H(v',x')$ bits of $y'$ equal to $1$. Then, $d_H(v',x')-\ell$ is the number of the first $d_H(v',x')$ bits of $y'$ equal to $0$. We obtain $d_H(v',y')-d_H(x',y')=\ell-(d_H(v',x')-\ell)=2\ell-d_H(v',x')$. Thus, adding $e'$ to $G_H$ shortens the distance from $v'$ to $y'$ by $2\ell-d_H(v',x')-1$.
	
	The number of nodes where exactly $\ell$ of the first $d_H(v',x')$ bits are equal to 1 is $\binom{d_H(v',x')}{\ell}\cdot 2^{d-d_H(v',x')}$. Therefore, we get a distance decrease for $v'$ of
	\begin{align*}
		&d_{G_H}(v',V_{G_H})-d_{G_H+e'}(v',V_{G_H})\\
		&=\sum_{\ell=\left\lceil\frac{d_H(v',x')}{2}\right\rceil+1}^{d_H(v',x')} \binom{d_H(v',x')}{l}\cdot 2^{d-d_H(v',x')}\cdot(2\ell-d_H(v',x')-1)\\
		&=2^{d-d_H(v',x')}\sum_{\ell=0}^{\left\lfloor\frac{d_H(v',x')}{2}\right\rfloor-1} \binom{d_H(v',x')}{\ell}\cdot(d_H(v',x')-2\ell-1)\\
		&\ge 2^{d-d_H(v',x')}\sum_{\ell=0}^{\left\lfloor\frac{d_H(v',x')}{2}\right\rfloor-1} \binom{d_H(v',x')}{\ell}\\
		&\ge 2^{d-d_H(v',x')}\frac{1}{2}\left(\sum_{\ell=0}^{d_H(v',x')}\binom{d_H(v',x')}{\ell}-\binom{d_H(v',x')}{\left\lfloor\frac{d_H(v',x')}{2}\right\rfloor}\right)\\
		&=2^{d-d_H(v',x')}\frac{1}{2}\left(2^{d_H(v',x')}-\binom{d_H(v',x')}{\left\lfloor\frac{d_H(v',x')}{2}\right\rfloor}\right)\\
		&\ge 2^{d-d_H(v',x')}\frac{1}{2}2^{d_H(v',x')-1}\\
		&=\frac{2^d}{4}.
	\end{align*}
	
	Observe that the distances from $v$ to all nodes in other cliques in $G$ are exactly the same as the distances from $v'$ to all other nodes in $G_H$. The same holds for $G+e$ and $G_H+e'$, with the exception of the distances from $v$ to the (at most 3) nodes in the same clique as $x$. 
	We distinguish two cases:

	If $2\cdot 2^d\le n<3\cdot 2^d$, each clique consists of 2 or 3 nodes. Therefore, we have a distance decrease of at least
	\begin{align*}
		d_G(v,V_{G})-d_{G+e}(v,V_{G})\ge 2 (d_{G_H}(v',V_{G_H})-d_{G_H+e'}(v',V_{G_H})) -3\ge\frac{2^d}{2}-3\ge \frac{n}{6}-3.
	\end{align*}
	If $3\cdot 2^d\le n<4\cdot 2^d$, each clique consists of 3 or 4 nodes. This means, we have a distance decrease of at least
	\begin{align*}
		d_G(v,V_{G})-d_{G+e}(v,V_{G})\ge 3 (d_{G_H}(v',V_{G_H})-d_{G_H+e'}(v',V_{G_H})) -3\ge 3\frac{2^d}{4}-3\ge\frac{n}{6}-3.
	\end{align*}
	Thus, edge additions are not beneficial for the incident agents if $\alpha\le\frac{n}{6}-3$ and we conclude that the constructed network is stable for $1\le\alpha\le\frac{n}{6}-3$. We also see that the number of edges $m$ is in $\Theta(n\log n)$ and the distance $d(V_G,V_G)$ is in $\Theta(n^2\log n)$. Since the social optimum for $1\le\alpha\le\frac{n}{6}-3$ is the path $P_n$, we get for the Price of Anarchy:
	\begin{equation*}
		\POA_n\ge\frac{SW(P_n)}{SW(G)}=\frac{\alpha (n-1)+\Theta(n^3)}{\alpha\Theta(n\log n)+\Theta(n^2\log n)}\in\Omega\left(\frac{n}{\log n}\right).\qedhere
	\end{equation*}
\end{proof}
\begin{proof}[Proof of (\ref{thm:com:poa:alpha_smaller_n_2})]
	We construct a path of cliques in the following way. Let $2\le d\le \frac{n-6}{2}$ be some even number and $c=\left\lfloor\frac{n-6}{d}\right\rfloor$. Furthermore, let $K_1,\dots,K_d$ be $d$ cliques consisting of $c$ or $c+1$ nodes, such that $\sum_{i=1}^d |V_{K_i}|=n-6$ and $\sum_{i=1}^{\frac{d}{2}}|V_{K_i}|=\left\lceil\frac{n-6}{2}\right\rceil$ and $\sum_{i=\frac{d}{2}+1}^{d}|V_{K_i}|=\left\lfloor\frac{n-6}{2}\right\rfloor$, and $v_1,v_1',v_2,v_2',v_3,v_3'$ be 6 more nodes. We now define the network $G$ as
	\begin{align*}
		V_G&\coloneqq\bigcup_{i=1}^d V_{K_i} \cup \{v_1,v_1',v_2,v_2',v_3,v_3'\},\\
		E_G&\coloneqq\bigcup_{i=1}^d E_{K_i} \cup\{\{v_1,v_1'\},\{v_2,v_2'\},\{v_3,v_3'\}\}\cup\{\{v,x\}\mid v\in\{v_2,v_2'\}\wedge x\in\{v_1,v_1',v_3,v_3'\}\}\\*
		&\hspace{0.6cm}\cup \bigcup_{i=1}^{\frac{d}{2}-1}\{\{v,x\}\mid v\in K_i\wedge x\in K_{i+1}\}\cup\bigcup_{v\in K_\frac{d}{2}}\{\{v,v_1\},\{v,v_1'\}\}\\*
		&\hspace{0.6cm}\cup\bigcup_{i=\frac{d}{2}+1}^{d-1}\{\{v,x\}\mid v\in K_i\wedge x\in K_{i+1}\}\cup\bigcup_{v\in K_{\frac{d}{2}+1}}\{\{v,v_3\},\{v,v_3'\}\}.
	\end{align*}
	\Cref{fig:path_of_cliques} shows a sketch of $G$.
	\begin{figure}[h]
		\centering
		\includegraphics[width=0.9\textwidth]{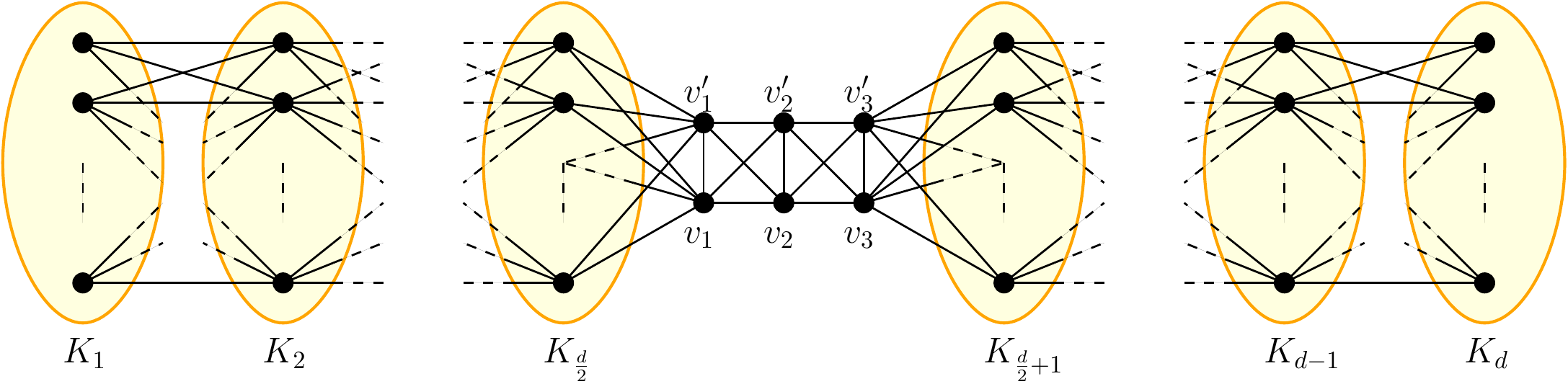}
		\caption{The figure shows clique network for a path consisting of $d$ cliques $K_1,\dots,K_d$ highlighted in yellow with 6 additional nodes in the middle. Note, that edges inside the cliques are not shown in this figure.}
		\label{fig:path_of_cliques}
	\end{figure}
	
	We observe that $G$ is stable against edge removal because of \Cref{clique_network}, since $\alpha>\frac{n}{6}-3\ge 1$ and $G$ being a clique network for the path. We now show that adding an edge is also not an improving move.
	
	We quickly see that, for a node $v$, adding an edge into the 2-neighborhood always shortens distances the least. We therefore only have to consider these edges. We observe that adding an edge between $v_1$ and $v_3$ (or because of symmetry, $v_1'$ or $v_3'$) decreases distances from $v_1$ to $v_3$ and all nodes in $K_{\frac{d}{2}+1},\dots,K_d$ and decreases distances from $v_3$ to $v_1$ and all nodes in $K_1,\dots,K_\frac{d}{2}$ by exactly 1. This means, we get
	\begin{align*}
		d_G(v_1,V)-d_{G+\{v_1,v_3\}}(v_1,V)&=1+\sum_{i=\frac{d}{2}+1}^d |V_{K_i}|=\left\lfloor\frac{n}{2}\right\rfloor-2\quad\text{and}\\
		d_G(v_3,V)-d_{G+\{v_1,v_3\}}(v_3,V)&=1+\sum_{i=1}^{\frac{d}{2}} |V_{K_i}|=\left\lceil\frac{n}{2}\right\rceil-2.
	\end{align*}
	
	Every other edge we could add decreases distances to all the cliques of one side of the path, resulting in larger distance decreases. This means that adding an edge is not an improving move for $\alpha\le\left\lfloor\frac{n}{2}\right\rfloor-2$. Therefore, $G$ is pairwise stable for the desired values of $\alpha$.
	
	We now evaluate the number of edges. We have $|E_{K_i}|\in\Theta(c^2)$. The number of edges between two neighboring cliques is also in $\Theta(c^2)$. This means that the total number of edges is $|E_G|\in\Theta(dc^2)$. We also see that the diameter of $G$ is $d$ and therefore $d_G(V,V)\in\mathcal{O}(dn^2)$. If we choose $d=2\left\lfloor\frac{\sqrt n}{2}\right\rfloor$, we have $d\in\Theta(\sqrt n)$ and $c\in \Theta(\sqrt n)$. Since $\alpha\in\Theta(n)$, we get for the Price of Anarchy
	\begin{equation*}
		\POA_n\ge\frac{\alpha(n-1)+\Theta(n^3)}{\alpha \Theta(dc^2)+\mathcal{O}(dn^2)}\in\Omega\left(\frac{n^3}{n^\frac{5}{2}+n^\frac{5}{2}}\right)=\Omega\left(\sqrt n\right).\qedhere
	\end{equation*}
\end{proof}
\begin{proof}[Proof of (\ref{thm:com:poa:largeAlpha})]
	This follows directly from the clique being socially optimal (see \Cref{thm:com:opt}) and the only pairwise stable network (see \Cref{thm:com:stable}).
\end{proof}
We have established that the Price of Anarchy is relatively high for $\alpha\le\frac{n}{2}$. It even meets the trivial upper bound of $\mathcal{O}(n)$ for a large range of $\alpha$. In contrast to the high $\POA$ values, we observe that the Price of Stability is independent of $\alpha$ and best possible.
\begin{restatable}{thm}{completePos}
	\label{complete:PoS}
	The Price of Stability is 1.
\end{restatable}
\begin{proof}
	This follows directly from the path being stable and socially optimal for $\alpha\le\frac{n}{3}$ and the clique being stable and socially optimal for $\alpha\ge\frac{n}{3}$ (see \Cref{thm:com:opt} and \Cref{thm:com:stable}).
\end{proof}

\noindent From an efficiency point-of-view, the huge gap between the PoA and the PoS suggests that having an outside force assigning an initial strategy to all players is beneficial. That way, stability and optimal social welfare can be guaranteed. Without such coordination, the players could end up in socially bad equilibria or in a cyclic sequence of improving moves.


\section{General Host Networks}
\label{sec:general}
We now analyze the SDNCG on arbitrary connected but not necessarily complete host networks. First, we analyze socially optimal networks and then we investigate the existence of pairwise stable networks. We prove our main result that establishes equilibrium existence on any connected host network for a wide parameter range of $\alpha$. Finally, we derive bounds on the Price of Anarchy and the Price of Stability.
Additionally, we show that computing the social optimum and the Maximum Routing-Cost Spanning Tree is NP-hard while computing a Swap-Maximal Routing-Cost Spanning Tree can be done in polynomial time.

\subsection{Stable and Optimal Networks}
\label{sec:general:stableOptimal}
While for the $K$-SDNCG, we only have two possible social optima (dependent on $\alpha$), this gets more complicated for general host networks. Of course, if they exist on general host networks, then the optima for the $K$-SDNCG are still the most efficient networks. Intuitively, if the host network does not contain a Hamilton path, then the social optimum should be a tree if $\alpha$ is small enough. Since all trees have the same number of edges, the social welfare of a tree is only influenced by the total distances. Remember that the spanning tree maximizing the total distances is by definition the Maximum Routing-Cost Spanning Tree (MRCST). We now show, that this intuition is indeed correct.
\begin{restatable}[Social Optimum]{thm}{generalOptimum}\label{thm:gen:opt}
	Let $H$ be a connected host network containing $n$ nodes.
	\begin{bracketenumerate}
		\item If $H$ contains a Hamilton path, then this path is the social optimum for $\alpha\le\frac{n}{3}$. The Hamilton path is the unique social optimum if $\alpha<\frac{n}{3}$.\label{thm:gen:opt:hamilton}
		\item For $\alpha\le1$, the MRCST of $H$ is socially optimal.\label{thm:gen:opt:maxrcst}
		\item For $\alpha>\frac{1}{24}(n-2)n(n+2)$, $H$ itself is the unique social optimum.\label{thm:gen:opt:large_alpha}
	\end{bracketenumerate}

\end{restatable}
\begin{proof}[Proof of (\ref{thm:gen:opt:hamilton})]
	This follows directly from \Cref{thm:com:opt}.
\end{proof}
\begin{proof}[Proof of (\ref{thm:gen:opt:maxrcst})]
	Let $T_{opt}$ be the MRCST of $H$ and $G$ some state of $H$, that is, a spanning subnetwork of $H$. Furthermore, let $T$ be a spanning tree of $G$. Then $\SW(G)\le\SW(T)$, since we can construct $G$ by adding edges to $T$ and for every edge added, the social welfare goes up by $\alpha\le1$ and down by at least 1, because of distances decreasing. Since $T_{opt}$ maximizes the total distances, we have $\SW(T)\le\SW(T_{opt})$ and therefore $\SW(G)\le\SW(T_{opt})$.
\end{proof}
\begin{proof}[Proof of (\ref{thm:gen:opt:large_alpha})]
	We show that any edge added to any network shortens the total distances by at most $\frac{1}{24}(n-1)n(n+1)$. It is easy to see that adding an edge between the two endpoints of a Hamilton path maximizes the distance decrease. This means, if $\alpha$ is larger than that, it is always socially better to add more edges to the network, resulting in $H$ itself to be optimal.
	
	We already know the social welfare of a Hamilton path $P_n$ from \Cref{sec:complete}. After adding an edge between the two endpoints of $P_n$, we get a cycle $C_n$. In this cycle, we see that for each node there are exactly two nodes for every possible distance $1\le d<\frac{n}{2}$ and an additional node at distance $\frac{n}{2}$, if $n$ is even. This yields
	\begin{align*}
		\SW(P_n)&=2\alpha(n-1)+\frac{1}{3}(n-1)n(n+1),\\
		\SW(C_n)&=2\alpha n+n\sum_{i=1}^{\frac{n-1}{2}}2i=2\alpha n+\frac{1}{4}(n-1)n(n+1)&&\text{for }n\text{ odd},\\
		\SW(C_n)&=2\alpha n+n\left(\sum_{i=1}^{\frac{n}{2}-1}2i+\frac{n}{2}\right)=2\alpha n+\frac{1}{4}(n-2)n^2+\frac{n^2}{2}&&\text{for }n\text{ even}.
	\end{align*}
	
	We then have $\SW(C_n)>\SW(P_n)$ if and only if
	\begin{align*}
		\alpha&>\frac{1}{24}(n-1)n(n+1)&&\text{for }n\text{ odd,}\\*
		\alpha&>\frac{1}{24}(n-2)n(n+2)>\frac{1}{24}(n-1)n(n+1)&&\text{for }n\text{ even.}
	\end{align*}
	This shows the claim.
\end{proof}
\noindent Contrasting statement $(3)$ from \Cref{thm:gen:opt}, we observe that for $\alpha<\frac{1}{24}(n-2)n(n+2)$, the host network is not necessarily the social optimum. Consider the host network $H\coloneqq C_n$ for even $n$, i.e., an even cycle with $n$ nodes. In the proof of (\ref{thm:gen:opt:large_alpha}), we see that $\SW(P_n)>\SW(C_n)$, implying that $C_n$ cannot be the social optimum. In fact, in this example, $P_n$ is the optimum since there are only two possible states (up to isomorphism): $P_n$ and $C_n$ itself. This is in stark contrast to the $K$-SDNCG, where the host network is optimal for $\alpha\ge\frac{n}{3}$.

Since finding a MRCST is NP-hard \cite{CAMERINI83}, finding the social optimum for a given host network must also be NP-hard.
\begin{restatable}[Computational Hardness]{thm}{generalNPHard}
	\label{thm:gen:np-hard}
	Computing the social optimum for a connected host network $H$ is NP-hard.
\end{restatable}
\begin{proof}
	This follows directly from the MRCST being hard to compute \cite{CAMERINI83} and the unique social optimum for $\alpha<1$ (see \Cref{thm:gen:opt}).
\end{proof}

\noindent Next, we discuss stable networks. In contrast to the $K$-SDCNG, it is not obvious that pairwise stable networks are guaranteed to exist for any connected host network. However, we can directly transfer the result that spanning trees are stable for small~$\alpha$. For large $\alpha$, similar to the clique being the unique stable network for $\alpha>\frac{n}{2}$ for complete host networks, as shown in \Cref{thm:com:stable}, we show that the whole host network is pairwise stable. However, in contrast to the $K$-SDNCG, this is true only for much larger values of $\alpha$.
\begin{restatable}[Stable Networks]{thm}{generalStable}\label{thm:gen:stable}
	Let $H$ be a connected host network containing $n$ nodes.
	\begin{bracketenumerate}
		\item For $\alpha\le 1$, every spanning tree of $H$ is pairwise stable. For $\alpha<1$, spanning trees are the only pairwise stable networks.\label{thm:gen:stable:tree}
		\item For $\alpha>\frac{1}{4}(n-1)^2$, $H$ is the only pairwise stable network.\label{thm:gen:stable:large_alpha}
	\end{bracketenumerate}

\end{restatable}
\begin{proof}[Proof of (\ref{thm:gen:stable:tree})]
	The proof is exactly the same as for (\ref{thm:com:stable:tree}) of \Cref{thm:com:stable}.
\end{proof}
\begin{proof}[Proof of (\ref{thm:gen:stable:large_alpha})]
	Consider a network $G$ on a host network $H=(V,E_H)$. The largest distance decrease a node $v\in V$ can suffer when forming an edge $e\in E_H$ is when $G$ is a path and $v$ one of its endpoints connecting to the other endpoint $x\in V$. This move decreases the distances by
	\begin{align*}
		\Delta d\coloneqq {}&d_{G}(v,V)-d_{G+e}(v,V)=d_G(x,V)-d_{G+e}(x,V)\\*
		={}&\sum_{i=1}^{\frac{n-1}{2}}(n-2i)=\frac{1}{4}(n-1)^2&&\text{for } n \text{ odd,}\\*
		\Delta d={}&\sum_{i=1}^{\frac{n}{2}-1}(n-2i)=\frac{1}{4}(n-2)n<\frac{1}{4}(n-1)^2&&\text{for }n\text{ even.}
	\end{align*}
	
	Thus, since $\alpha>\frac{1}{4}(n-1)^2$, forming edges is always beneficial for the incident nodes. Similarly, edge removal always decreases the utility of the incident nodes. Therefore, the host network $H$ is the only pairwise stable network.
\end{proof}

\noindent Contrasting statement (2) of \Cref{thm:gen:stable}, using $H\coloneqq C_n$ for odd $n$ and $\alpha<\frac{1}{4}(n-1)^2$ shows that the host network is not necessarily pairwise stable. 
This example also shows that the optimum is not necessarily stable: For $\alpha\ge\frac{1}{4}(n-1)^2$ and $H\coloneqq C_n$ as the host network, $C_n$ is the only pairwise stable network but it is not the optimum for $\alpha<\frac{1}{24}(n-2)n(n+2)$. This is another significant difference to the $K$-SDNCG.

Now that we characterized stable networks for extreme $\alpha$-values, the question remains whether stable states also exist for in-between values. For the $K$-SDNCG, the path is stable up to $\alpha<\frac{n-1}{2}$. This is, of course, still true for non-complete host networks if they contain a Hamilton path. Since a Hamilton path (if it exists) is the MRCST, it is natural to suspect that the MRCST properties at least partially ensure stability for some $\alpha \geq 1$. However, even if true, the MRCST is still NP-hard to compute. Hence, in quest of an efficiently computable stable network, we introduce a less strict variant of MRCSTs which is only locally optimal: Swap-Maximal Routing-Cost Spanning Trees. Remember, a SMRCST is a spanning tree whose summed distances cannot be increased by removing one edge and adding another edge.

As our main result, we now show that SMRCSTs (and therefore MRCSTs, too) are indeed stable beyond $\alpha\le 1$. Note, that for the inverse model of the NCG on an arbitrary host network~\cite{demaine2009}, so far no equilibrium existence statement is known.
\begin{restatable}{thm}{hugeProof}\label{general:stable:huge_proof}
	Let $H$ be a connected host network containing $n$ nodes. Then for $\alpha\le \frac{n}{3}$, any Swap-Maximal Routing-Cost Spanning Tree is pairwise stable.
\end{restatable}
\begin{proof}
	Let $G$ with $V_G=V_H$ and $E_G\subseteq E_H$ be a SMRCST. Since $G$ is a tree, we only have to consider edge additions. We show that adding any edge decreases the summed distances for at least one of the edges endpoints by at least $\frac{n}{3}$. This is sufficient to show the claim.
	
	Let $e_1\in E_H\setminus E_G$ be an edge not part of the SMRCST. Adding $e_1$ would form a cycle of length $d\in\N$ consisting of nodes $v_1,\dots, v_d\in V$ with $v_1$ and $v_d$ being the nodes incident to~$e_1$.
	Let $E_C$ be the set of all edges on this cycle. Removing all edges in $E_C$ from $G$ would create $d$ trees rooted in $v_1,\dots,v_d$ respectively. Let furthermore $x_1,\dots,x_d$ be the number of nodes in each of the $d$ trees. See \Cref{app:fig:maxrcst} for an illustration.
	\begin{figure}
		\centering
		\includegraphics[width=0.37\textwidth]{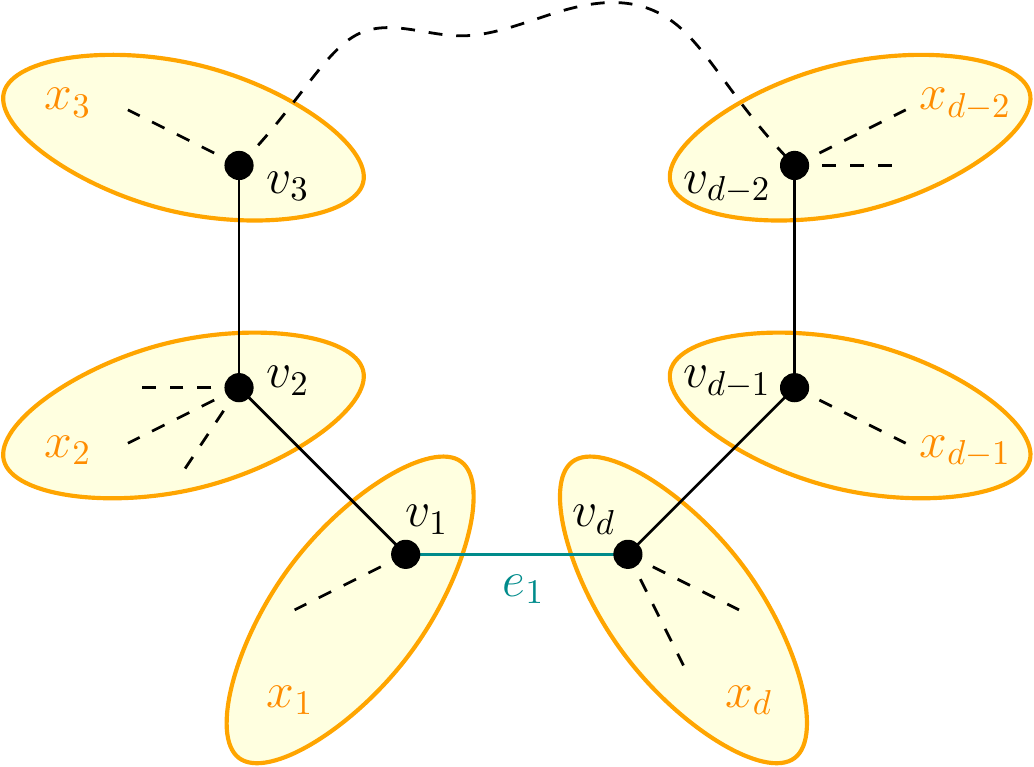}
		\includegraphics[width=0.27\textwidth]{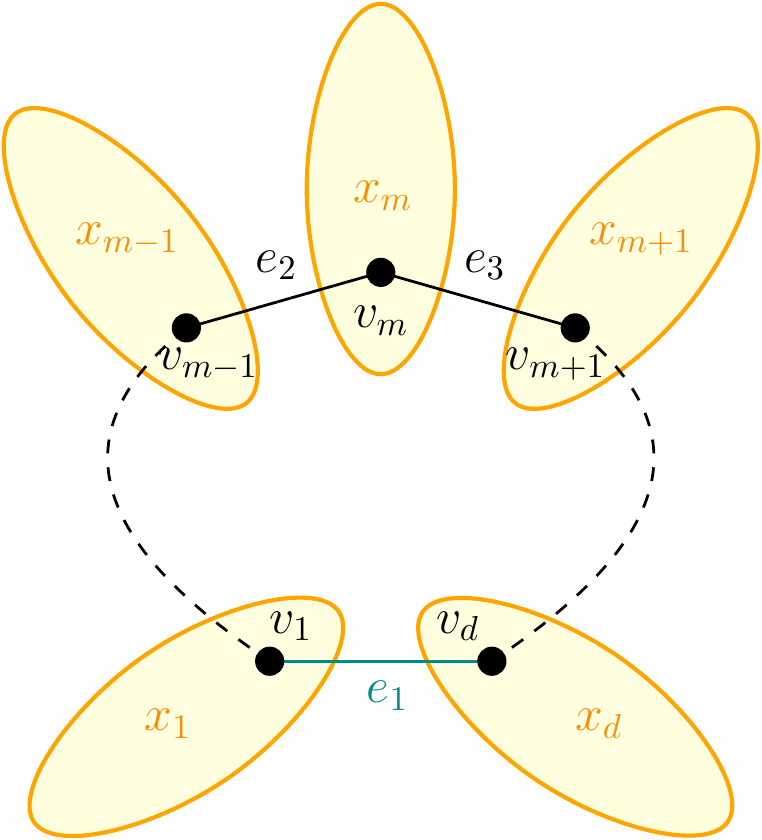}
		\includegraphics[width=0.34\textwidth]{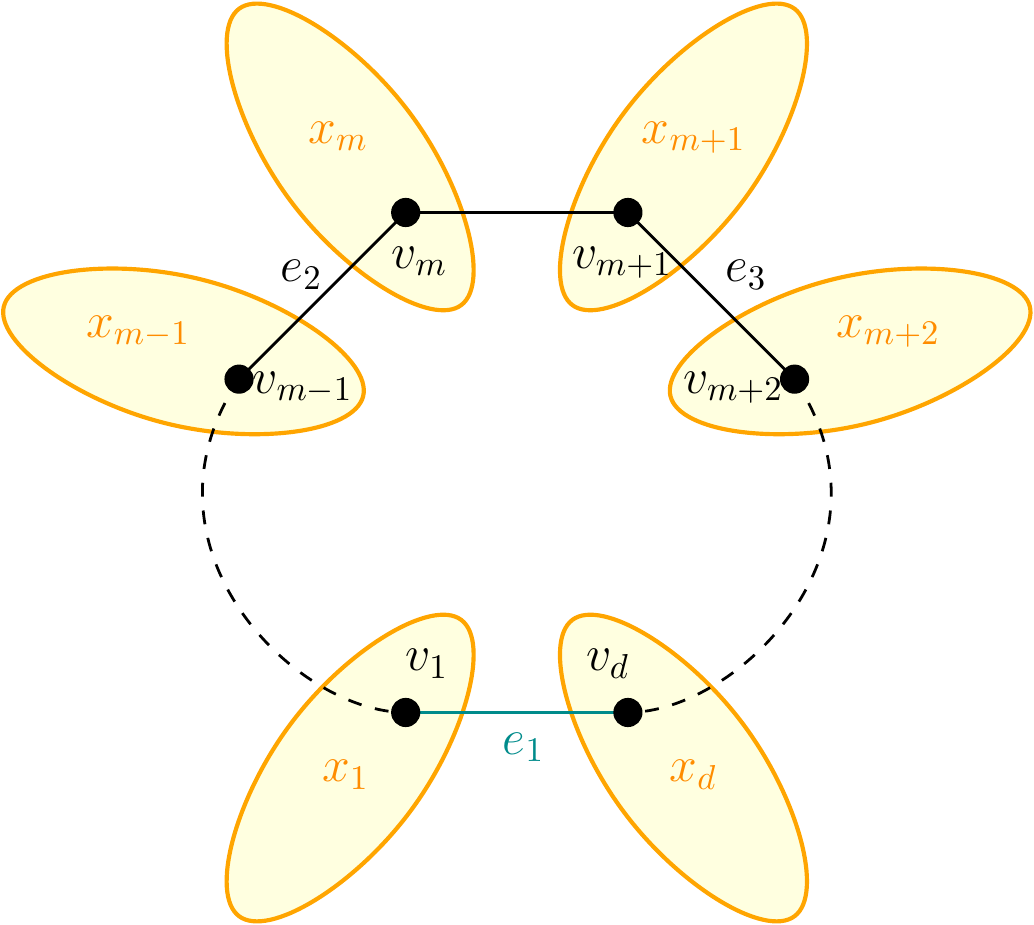}
		\caption{This figure shows the cycle formed by adding $e_1$ to the SMRCST. The cycle is of length $d$ and contains the nodes $v_1,\dots,v_d$. Every other node is contained in one of the subtrees rooted in one of the nodes on the cycle. These subtrees are represented in yellow. The number of nodes contained in the subtree rooted in $v_i$ is $x_i$. Middle and right: the cycle for $d$ being odd or even, respectively, and the two special edges $e_2$ and $e_3$.}
		\label{app:fig:maxrcst}
	\end{figure}
	
	Since $G$ is a tree, there is exactly one path between each pair of nodes (which is also the shortest). For each edge $e\in E_G$, we define $d_G(e)$ as the number of paths between two nodes in $G$ which include $e$. We then can express the total distances as
	\begin{align*}
		d_G(V,V)=2\sum_{e\in E_G}d_G(e).
	\end{align*}
	Note, that each path between two nodes contributes twice to the total distances (one for each node), which leads to the factor of 2.
	
	Let $x\coloneqq (x_1,\dots,x_d)$.
	We now define for each edge $e\in E_C$ on the cycle
	\begin{align*}
		c_e(x)\coloneqq \sum_{e'\in E_C}d_{G+e_1-e}(e')=\sum_{i=1}^{d-1}\sum_{j=i+1}^d x_i x_j d_{G+e_1-e}(v_i,v_j).
	\end{align*}
	This is the contribution of all the edges on the cycle to the total distances if we add $e_1$ to it and instead remove $e$ from it.
	Note that $c_{e_1}$ is the value for the original network since $G+e_1-e_1=G$.
	We see that $c_e$ does not depent on the structure of the subtrees rooted in the $v_i$ but only on the number of nodes in each subtree.
	Since the number of paths going over an edge that is not on the cycle does not change when we add $e_1$ and remove $e$, we have
	\begin{align*}
		&d_G(V,V)-d_{G+e-e_1}(V,V)\\
		&=2\sum_{e'\in E}(d_G(e')-d_{G+e-e_1}(e'))\\
		&=2\sum_{e'\in E_C}(d_G(e')-d_{G+e-e_1}(e'))+2\sum_{e'\in E_G\setminus E_C}(d_G(e')-d_{G+e-e_1}(e'))\\
		&=2c_{e_1}(x)-2c_e(x).
	\end{align*}
	
	We know that $G$ is a SMRCST of $H$. This means that $d_G(V,V)\ge d_{G'}(V,V)$ for any other spanning tree $G'$ which can be obtained from $G$ by a swap of one edge. We therefore also have
	\begin{align}
		\label{app:general:stable:constr}
		\forall e\in E_C\colon c_{e_1}(x)\ge c_e(x).
	\end{align}
	
	Now, we use the previous observations to formulate and solve a minimization problem which yields the desired bound. We start with some definitions.
	
	We call $x=(x_1,\dots,x_d)\in \N^d$ with $x_i\ge 1$ and $\sum_{i=1}^d x_i=n$ a \emph{node distribution}. For each edge $e\in E_C$, we call $c_e(x)$ (defined above) the \emph{cost} of $e$. And lastly, we define the \emph{distance decrease} $\Delta d$ as
	\begin{align}
		\Delta d(x)\coloneqq \max \left\{\sum_{i=1}^{\left\lfloor\frac{d-1}{2}\right\rfloor}(d-2i)x_i,\sum_{i=1}^{\left\lfloor\frac{d-1}{2}\right\rfloor}(d-2i)x_{d-i+1}\right\}.
	\end{align}
	
	The goal then is: Find a node distribution $x$ that fulfills \ref{app:general:stable:constr} and minimizes $\Delta d(x)$. Observe that this indeed yields a lower bound for the distance decrease when adding $e$ to $G$. If we show that this is at least $\frac{n}{3}$, we proved the statement.
	
	Let $x=(x_1,\dots,x_d)\in\N^d$ be a node distribution minimizing $\Delta d(x)$. We first show the claim for $d=3$ and $d=4$.
	
	For $d=3$, we have $x_2\le x_1$ and $x_2\le x_3$ from \ref{app:general:stable:constr} and $\Delta d(x)=\max\{x_1,x_3\}$. Since $x_1+x_2+x_3=n$, this yields $\Delta d(x)\ge \frac{n}{3}$.
	
	For $d=4$, we have $c_{e_1}(x)-c_{\{x_1,x_2\}}(x)=2x_1x_4-2x_1x_2$, and therefore $c_{e_1}\ge c_{\{x_1,x_2\}}$ if and only if $x_4\ge x_2$. Similarly, we get $x_1\ge x_3$ from $\{x_3,x_4\}$. Together with $x_1+x_2+x_3+x_4=n$, we see that $\max\{x_1,x_4\}\ge\frac{n}{4}$. We conclude that $\Delta d(x)=2\max\{x_1,x_4\}\ge\frac{n}{2}>\frac{n}{3}$.
	
	For $d>4$, we make a case distinction between $d$ being odd and $d$ being even and simplify the problem by doing several relaxation steps. For $d\le 4$, it is easy to show that $\Delta d(x)\ge\frac{n}{3}$. For further steps, we allow $x\in\R_{\ge 1}^d$. Note, that this only allows for smaller minima and therefore still yields a lower bound for the original problem.
	
	The high level idea of the following steps is that we can redistribute weights of the node distribution $x$ without changing $\Delta(x)$ or violating \ref{app:general:stable:constr} and thereby reducing the number of variables contained in $x$ by setting most $x_i$ to 1. We now make a case distinction.
	
	\proofsubparagraph{Case $d$ odd:}
	Let $m=\frac{d+1}{2}$ and $e_2\coloneqq \{v_{m-1},v_m\}$ and $e_3\coloneqq \{v_m,v_{m+1}\}$. Thus, $v_m$ is the node equidistant from $v_1$ and $v_d$ in $C$ and $e_2$ and $e_3$ are the edges on $C$ incident to $v_m$. (see \Cref{app:fig:maxrcst}~(middle)) We will only consider the two constraints
	\begin{align}
		\label{app:general:stable:constr_odd}
		c_{e_1}(x)\ge c_{e_2}(x)\text{ ~~~~~and~~~~~ } c_{e_1}(x)\ge c_{e_3}(x),
	\end{align}
	where
	\begin{align*}
		c_{e_1}(x)&=\sum_{i=1}^{d-1}\sum_{j=i+1}^d (j-i) x_i x_j,\\
		c_{e_2}(x)&=\sum_{i=1}^{m-2}\sum_{j=i+1}^{m-1}(j-i)x_i x_j + \sum_{i=m}^{d-1}\sum_{j=i+1}^d (j-i)x_i x_j + \sum_{i=1}^{m-1}\sum_{j=m}^d (i+d-j)x_i x_j,\\
		c_{e_3}(x)&=\sum_{i=1}^{m-1}\sum_{j=i+1}^{m}(j-i)x_i x_j + \sum_{i=m+1}^{d-1}\sum_{j=i+1}^d (j-i)x_i x_j + \sum_{i=1}^{m}\sum_{j=m+1}^d (i+d-j)x_i x_j.
	\end{align*}
	This still yields a lower bound for the original problem since the constraints from \ref{app:general:stable:constr_odd} are a subset of the constraints from \ref{app:general:stable:constr}.
	
	We observe that the claim is trivially true for $\frac{n}{3}\le\frac{d+3}{2}$ since $\Delta d(x)\ge\frac{d+3}{2}$ (tight for $n=5$, $d=5$, $x=(1,1,1,1,1)$). This means, we can assume
	\begin{equation}
		\label{app:general:odd:eq:n_bigger_3m}
		n>\frac{3(d+3)}{2}, \text{ ~~~~~or rather~~~~~ } n>3(m+1).
	\end{equation}
	
	In the following, we perform a series of relaxations. For that, we define the node distributions $\x1,\x2,\x3$ such that for all $1\le i\le d$
	\begin{align*}
		\x1_i &\coloneqq \begin{cases}
			x_1 + \sum_{p=2}^{m-2}\frac{m-1-p}{m-2}(x_p-1)&\text{if }i=1,\\
			x_{m-1}+\sum_{p=2}^{m-2}\frac{p-1}{m-2}(x_p-1)&\text{if }i={m-1},\\
			x_m&\text{if }i=m,\\
			x_{m+1}+\sum_{p=m+2}^{d-1}\frac{d-p}{m-2}(x_p-1)&\text{if }i=m+1,\\
			x_d + \sum_{p=m+2}^{d-1}\frac{p-m-1}{m-2}(x_p-1)&\text{if }i=d,\\
			1&\text{else,}
		\end{cases}\\
		&\qquad\text{(Without loss of generality, we assume } \x1_1\ge \x1_d\\
		&\qquad\text{and } (2m-3)\x1_1+\x1_{m-1}=(2m-3)\x1_d+\x1_{m+1}.)\\
		\x2_i &\coloneqq\begin{cases}
			\x1_1&\text{if }i=1\text{ or } i=d,\\
			\x1_{m-1}&\text{if }i=m-1\text{ or } i=m+1,\\
			\x1_m+(2m-4)\left(\x1_1-\x1_d\right)&\text{if }i=m,\\
			1&\text{else,}
		\end{cases}\\
		\x3_i &\coloneqq\begin{cases}
			\x2_{m}+2\frac{2m-4}{2m-3}\left(\x2_{m-1}-1\right)&\text{if } i=m,\\
			\x2_1+\frac{1}{2m-3}\left(\x2_{m-1}-1\right)&\text{if } i=1 \text{ or } i=d,\\
			1&\text{else}.
		\end{cases}.
	\end{align*}
	
	It is easy to see that $\Delta d(x)=\Delta d\left(\x1\right)=\Delta d\left(\x2\right)=\Delta d\left(\x3\right)$. We show that $c_{e_1}\left(\x{i}\right)\ge c_{e_2}\left(\x{i}\right)$ and $c_{e_1}\left(\x{i}\right)\ge c_{e_3}\left(\x{i}\right)$ for $1\le i\le 3$. This means, $\x3$ is also a solution of the minimization problem.
	
	First, we show this for $\x1$. Let $1<p<m-1$. Consider $x^*$, a modification of $x$ where the weight of $x_p$ is distributed among $x_1$ and $x_{m-1}$ as follows:
	\begin{align*}
		\forall 1\le i\le d\colon x^*_i = \begin{cases}
			x_1 + \frac{m-1-p}{m-2}(x_p-1)&\text{if }i=1,\\
			1&\text{if }i=p,\\
			x_{m-1}+\frac{p-1}{m-2}(x_p-1)&\text{if }i={m-1},\\
			x_i&\text{else.}
		\end{cases}
	\end{align*}
	To show that $x^*$ also fulfills (\ref{app:general:stable:constr_odd}), we show that $c_{e_2}(x)-c_{e_1}(x)=c_{e_2}(x^*)-c_{e_1}(x^*)$ and $c_{e_3}(x)-c_{e_1}(x)=c_{e_3}(x^*)-c_{e_1}(x^*)$. We have
	\begin{align*}
		c_{e_2}(x)-c_{e_1}(x)=\sum_{i=1}^{m-1}\sum_{j=m}^d (i+d-j-(j-i))x_i x_j=\sum_{i=1}^{m-1}\sum_{j=m}^d (2i+d-2j)x_i x_j
	\end{align*}
	and
	\begin{align*}
		&c_{e_2}(x^*)-c_{e_1}(x^*)\\
		={}&\sum_{i=1}^{m-1}\sum_{j=m}^d (2i+d-2j)x^*_i x^*_j\\
		={}&\sum_{i=1}^{m-1}\sum_{j=m}^d (2i+d-2j)x_i x_j+\frac{m-1-p}{m-2}(x_p-1)\sum_{j=m}^d(2+d-2j)x_j\\
		&-(x_p-1)\sum_{j=m}^d(2p+d-2j)x_j+\frac{p-1}{m-2}(x_p-1)\sum_{j=m}^d(2(m-1)+d-2j)x_j\\
		=&\sum_{i=1}^{m-1}\sum_{j=m}^d (2i+d-2j)x_i x_j+\frac{m-1-p}{m-2}(x_p-1)(d-m+1)2\\
		&-(x_p-1)(d-m+1)2p+\frac{p-1}{m-2}(x_p-1)(d-m+1)2(m-1)\\
		={}&\sum_{i=1}^{m-1}\sum_{j=m}^d (2i+d-2j)x_i x_j\\
		&+(x_p-1)(d-m+1)(-2p+\frac{2}{m-2}((m-1-p) + (p-1)(m-1)))\\
		={}&\sum_{i=1}^{m-1}\sum_{j=m}^d (2i+d-2j)x_i x_j+(x_p-1)(d-m+1)\left(-2p+\frac{2}{m-2}(m-2)p\right)\\
		={}&c_{e_2}(x)-c_{e_1}(x).
	\end{align*}
	The calculations for $c_{e_3}(x)-c_{e_1}(x)=c_{e_3}(x^*)-c_{e_1}(x^*)$ are exactly the same with the only difference being the sum indices (first sum goes to $m$ and second sum starts at $m+1$).
	
	Because of symmetry, for $m+1<p<d$, we can similarly distribute weights from $x_p$ to $x_{m+1}$ and $x_d$. Using this for all $1<p<m-1$ and $m+1<p<d$ iteratively, we get exactly $\x1$, which therefore still fulfills (\ref{app:general:stable:constr_odd}).
	
	Next, we show that $\x2$ fulfills (\ref{app:general:stable:constr_odd}), too. We have
	\begin{align*}
		c_{e_1}(\x1)={}&(m-2)\x1_1\x1_{m-1}+(m-1)\x1_1\x1_m+m\x1_1\x1_{m+1}+(2m-2)\x1_1\x1_d\\*&
		+\x1_{m-1}\x1_m+2\x1_{m-1}\x1_{m+1}+m\x1_{m-1}\x1_d+\x1_m\x1_{m+1}+(m-1)\x1_m\x1_d\\*
		&+(m-2)\x1_{m+1}\x1_d+\x1_1(m-3)(2m-2)+\x1_d(m-3)(2m-2)\\*
		&+\x1_{m-1}(m-3)m+\x1_{m+1}(m-3)m+\x1_m(m-3)m\\*
		&+\frac{1}{3}(m-4)(m-3)(m-2)+(m-3)(m-3)m,\\
		c_{e_2}(\x1)={}&(m-2)\x1_1\x1_{m-1}+m\x1_1\x1_m+(m-1)\x1_1\x1_{m+1}+\x1_1\x1_d\\*
		&+(2m-2)\x1_{m-1}\x1_m+(2m-3)\x1_{m-1}\x1_{m+1}+(m-1)\x1_{m-1}\x1_d+\x1_m\x1_{m+1}\\*
		&+(m-1)\x1_m\x1_d+(m-2)\x1_{m+1}\x1_d\\*
		&+\x1_1(m-3)(m-1)+\x1_d(m-3)(m-1)+\x1_{m-1}(m-3)(2m-3)\\*
		&+\x1_{m+1}(m-3)(2m-3)+\x1_m(m-3)(2m-1)\\*
		&+\frac{1}{3}(m-4)(m-3)(m-2)+(m-3)(m-3)(m-1)\text{ and}\\
		c_{e_3}(\x1)={}&(m-2)\x1_1\x1_{m-1}+(m-1)\x1_1\x1_m+(m-1)\x1_1\x1_{m+1}+\x1_1\x1_d\\*
		&+\x1_{m-1}\x1_m+(2m-3)\x1_{m-1}\x1_{m+1}+(m-1)\x1_{m-1}\x1_d+(2m-2)\x1_m\x1_{m+1}\\*
		&+m\x1_m\x1_d+(m-2)\x1_{m+1}\x1_d\\*
		&+\x1_1(m-3)(m-1)+\x1_d(m-3)(m-1)+\x1_{m-1}(m-3)(2m-3)\\*
		&+\x1_{m+1}(m-3)(2m-3)+\x1_m(m-3)(2m-1)\\*
		&+\frac{1}{3}(m-4)(m-3)(m-2)+(m-3)(m-3)(m-1).
	\end{align*}
	We can further assume that $(2m-3)\x1_1+\x1_{m-1}=(2m-3)\x1_d+\x1_{m+1}$. This is because if (without loss of generality) $(2m-3)\x1_1+\x1_{m-1}<(2m-3)\x1_d+\x1_{m+1}$ holds, we can move value from $\x1_{m-1}$ to $\x1_1$ without changing $\Delta d(\x1)$. This increases $c_{e_1}(\x1)$ more than $c_{e_2}(\x1)$ and $c_{e_3}(\x1)$.
	
	Let (without loss of generality) $\x1_1\ge \x1_d$ and $y\coloneqq \x1_1-\x1_d$. Then we have $\x1_d=\x1_1-y$ and $\x1_{m+1}=\x1_{m-1}+(2m-3)y$. We see, that
	\begin{align*}
		\forall 1\le i\le d\colon \x2_i=\begin{cases}
			\x1_{m}+(2m-4)y&\text{if } i=m,\\
			\x1_{m-1}&\text{if } i=m-1\text{ or }i=m+1,\\
			\x1_1&\text{if } i=1\text{ or } i=d,\\
			1&\text{else}.
		\end{cases}
	\end{align*}
	We show $c_{e_1}(\x2)\ge c_{e_3}(\x2)$, by showing that $c_{e_1}(\x1)-c_{e_3}(\x1)-c_{e_1}(\x2)+c_{e_3}(\x2)>0$.
	
	We have
	\begin{align*}
		\Delta \x1\coloneqq {}& c_{e_1}(\x1)-c_{e_3}(\x1)\\
		={}&\x1_1\x1_{m+1}+(2m-3)\x1_1\x1_d-(2m-5)\x1_{m-1}\x1_{m+1}+\x1_{m-1}\x1_d\\*
		&-(2m-3)\x1_m\x1_{m+1}-\x1_m\x1_d+\x1_1(m-3)(m-1)+\x1_d(m-3)(m-1)\\*
		&-\x1_{m-1}(m-3)(m-3)-\x1_{m+1}(m-3)(m-3)-\x1_m(m-3)(m-1)\\*
		&+(m-3)(m-3)\\
		={}&\x1_1\x1_{m-1}+\x1_1(2m-3)y+(2m-3)\x1_1\x1_1-(2m-3)\x1_1y\\*
		&-(2m-5)\x1_{m-1}\x1_{m-1}-(2m-5)(2m-3)\x1_{m-1}y+\x1_{m-1}\x1_1-\x1_{m-1}y\\*
		&-(2m-3)\x1_m\x1_{m-1}-(2m-3)\x1_m(2m-3)y-\x1_m\x1_1+\x1_my\\*
		&+\x1_1(m-3)(m-1)+\x1_1(m-3)(m-1)-y(m-3)(m-1)\\*
		&-\x1_{m-1}(m-3)(m-3)-\x1_{m-1}(m-3)(m-3)\\*
		&-(2m-3)y(m-3)(m-3)-\x1_m(m-3)(m-1)+(m-3)(m-3)
		\shortintertext{and}
		\Delta \x2\coloneqq {}&c_{e_1}(\x2)-c_{e_3}(\x2)\\
		={}&\x1_1\x1_{m-1}+(2m-3)\x1_1\x1_1-(2m-5)\x1_{m-1}\x1_{m-1}+\x1_{m-1}\x1_1\\*
		&-(2m-3)\x1_m\x1_{m-1}-(2m-3)(2m-4)y\x1_{m-1}-\x1_m\x1_1-(2m-4)y\x1_1\\*
		&+\x1_1(m-3)(m-1)+\x1_1(m-3)(m-1)-\x1_{m-1}(m-3)(m-3)\\*
		&-\x1_{m-1}(m-3)(m-3)-\x1_m(m-3)(m-1)-(2m-4)y(m-3)(m-1)\\*
		&+(m-3)(m-3),
	\end{align*}
	and therefore
	\begin{align*}
		\Delta \x2-\Delta \x1={}&(2m-5)(2m-3)\x1_{m-1}y+\x1_{m-1}y+(2m-3)\x1_m(2m-3)y\\*
		&-(2m-3)(2m-4)y\x1_{m-1}-\x1_my-(2m-4)y\x1_1+y(m-3)(m-1)\\*
		&+(2m-3)y(m-3)(m-3)-(2m-4)y(m-3)(m-1)\\
		={}&\x1_{m-1}y((2m-5)(2m-3)+1-(2m-3)(2m-4))\\*
		&+\x1_my((2m-3)(2m-3)-1)-(2m-4)\x1_1y\\*
		&+y(m-3)(m-1+(2m-3)(m-3)-(2m-4)(m-1))\\
		={}&-2\x1_{m-1}y(m-2)+4\x1_my(m-2)(m-1)-2\x1_1y(m-2)\\*
		&-2y(m-3)(m-2)\\
		={}&2y(m-2)(-\x1_{m-1}+2\x1_m(m-1)-\x1_1-m+3)\\
		={}&2y(m-2)(\x1_m-\x1_1-\x1_{m-1}+2\x1_m(m-1.5)-m+3)\\
		\ge{}& 0.
	\end{align*}
	The last step follows from $\x1_m\geq \x1_1+\x1_{m-1}\ge 1$, $m\ge 3$ and $y\ge 0$.
	
	Since $\x1$ fulfills (\ref{app:general:stable:constr_odd}), we see that $c_{e_1}(\x2)\ge c_{e_3}(\x2)$ holds, too. With $\x2$ being symmetric, $c_{e_1}(\x2)\ge c_{e_2}(\x2)$ holds as well. This shows that $\x2$ also fulfills (\ref{app:general:stable:constr_odd}).
	
	Now, we show that $\x3$ fulfills (\ref{app:general:stable:constr_odd}), too. Let $y\coloneqq \frac{\x2_{m-1}-1}{2m-3}$. We see that
	\begin{align*}
		\forall 1\le i\le d\colon \x3_i=\begin{cases}
			\x2_{m}+2(2m-4)y&\text{if } i=m,\\
			\x2_1+y&\text{if } i=1 \text{ or } i=d,\\
			1&\text{else}.
		\end{cases}
	\end{align*}
	
	We see that
	\begin{align*}
		\Delta \x2={}&2\x2_1 + 2(2m-3)y\x2_1 + (2m-3)\x2_1\x2_1-(2m-5)(1+(2m-3)y)^2\\*
		&-(2m-3)\x2_m-(2m-3)\x2_m(2m-3)y-\x2_1\x2_m+2\x2_1(m-3)(m-1)\\*
		&-2(m-3)(m-3)-2(2m-3)y(m-3)(m-3)-\x2_m(m-3)(m-1)\\*
		&+(m-3)(m-3)\\
		={}&2\x2_1+(2m-3)\x2_1\x2_1-(2m-5)-(2m-3)\x2_m-\x2_1\x2_m\\*
		&+2\x2_1(m-3)(m-1)-2(m-3)(m-3)-\x2_m(m-3)(m-1)\\*
		&+(m-3)(m-3)\\*
		&+y(2m-3)(2\x2_1-(2m-5)2-(2m-3)\x2_m-2(m-3)(m-3))\\*
		&-y^2(2m-3)^2(2m-5),\\
		\Delta \x3\coloneqq {}&c_{e_1}(\x3)-c_{e_2}(\x3)\\
		={}&2\x2_1+2y+(2m-3)(\x2_1+y)^2-(2m-5)-(2m-3)(\x2_m+2(2m-4)y)\\*
		&-(\x2_1+y)(\x2_m+2(2m-4)y)+2(\x2_1+y)(m-3)(m-1)\\*
		&-2(m-3)(m-3)-(\x2_m+2(2m-4)y)(m-3)(m-1)\\*
		&+(m-3)(m-3)\\
		={}&2\x2_1+(2m-3)\x2_1\x2_1-(2m-5)-(2m-3)\x2_m-\x2_1\x2_m\\*
		&+2\x2_1(m-3)(m-1)-2(m-3)(m-3)-\x2_m(m-3)(m-1)\\*
		&+(m-3)(m-3)\\*
		&+y(2+(2m-3)2\x2_1-(2m-3)2(2m-4)-\x2_m-2(2m-4)\x2_1\\*
		&\qquad+2(m-3)(m-1)-2(2m-4)(m-3)(m-1))\\*
		&+y^2((2m-3)-2(2m-4))
	\end{align*}
	and obtain
	\begin{align*}
		\Delta \x3-\Delta \x2={}&y(2-2(2m-3)+(2m-3)^2\x2_m+2(2m-3)(m-3)(m-3)-\x2_m\\*
		&\quad-2(2m-4)\x2_1+2(m-3)(m-1)-2(2m-4)(m-3)(m-1))\\*
		&+y^2((2m-3)-2(2m-4)+(2m-3)^2(2m-5))\\
		={}&y(m-2)(-4\x2_1+4(m-1)\x2_m-4(m-2))\\*
		&+4y^2(2m-5)(m-1)(m-2)\\
		\ge{}&0.
	\end{align*}
	The last step follows from $\x2_m\geq \x2_1$, $m\ge 3$ and $y\ge 0$.
	
	Since $\x2$ fulfills (\ref{app:general:stable:constr_odd}), $\x3$ does, too. Because of $\x3$ being a node distribution with only two variables ($\x3_1=\x3_d$ and $\x3_m$), we can simplify $\Delta \x3$ and $\Delta d$ to
	\begin{align*}
		\Delta \x3 ={}& c_{e_1}(\x3)-c_{e_2}(\x3)\\
		={}&2\x3_1+(2m-3)\x3_1\x3_1-(2m-5)-(2m-3)\x3_m-\x3_1\x3_m\\*
		&+2\x3_1(m-3)(m-1)-2(m-3)(m-3)-\x3_m(m-3)(m-1)\\*
		&+(m-3)(m-3)\\
		={}&(2m-3)\x3_1\x3_1+2\x3_1(m-2)^2-\x3_1\x3_m-m(m-2)\x3_m-(m-2)^2
		\intertext{and}
		\Delta d(\x3)={}&(m-2)^2+(2m-3)\x3_1.
	\end{align*}
	
	We now prove the lower bound.
	
	We see that $\Delta d(\x3)$ is only dependent on $\x3_1$ which means we have to minimize $\x3_1$. Since $\x3$ sums to $n$, we have $n=2\x3_1+\x3_m+2m-4$, and therefore $\x3_m=n-2\x3_1-2m+4$.
	
	Substituting this into $\Delta \x3$ yields
	\begin{align*}
		\Delta \x3 ={}&(2m-3)\x3_1\x3_1+2\x3_1(m-2)^2-\x3_1(n-2\x3_1-2m+4)\\*
		&-m(m-2)(n-2\x3_1-2m+4)-(m-2)^2\\
		={}&\x3_1\x3_1(2m-1) + \x3_1(2(m-2)(2m-1)-n)+(2m-1)(m-2)^2\\*
		&-nm(m-2).
	\end{align*}
	
	Observe that $\x3$ fulfills (\ref{app:general:stable:constr_odd}) if and only if $\Delta \x3\ge 0$. Solving the quadratic equation, we see that this is only the case for
	\begin{align*}
		\x3_1\ge{}& -(m-2)+\frac{n}{2(2m-1)}\\*
		&+\sqrt{\left(\frac{n}{2(2m-1)}-(m-2)\right)^2+nm(m-2)-(2m-1)(m-2)^2}.
	\end{align*}
	With (\ref{app:general:odd:eq:n_bigger_3m}), we see that the term under the root is larger than
	\begin{align*}
		&0^2+3(m+1)m(m-2)-(2m-1)(m-2)^2\\
		\ge{}&(m+1)m(m-2)\\
		\ge{}&(m-2)^2.
	\end{align*}
	This yields
	\begin{align*}
		\x3_1\ge -(m-2)+\frac{n}{2(2m-1)}+\sqrt{(m-2)^2}=\frac{n}{2(2m-1)},
	\end{align*}
	and therefore
	\begin{align*}
		\Delta d(\x3)\ge(m-2)^2+\frac{2m-3}{2m-1}\cdot\frac{n}{2}\ge\frac{n}{3}.
	\end{align*}
	
	\proofsubparagraph{Case $d$ even:}
	Let $m=\frac{d}{2}$ and $e_2\coloneqq\{v_{m-1},v_m\}$ and $e_3\coloneqq\{v_{m+1},v_{m+2}\}$. (see \Cref{app:fig:maxrcst}~(right)) Again, we will only consider the two constraints
	\begin{align}
		\label{app:general:stable:constr_even}
		c_{e_1}(x)\ge c_{e_2}(x)\text{ ~~~~~and~~~~~ } c_{e_1}(x)\ge c_{e_3}(x),
	\end{align}
	where
	\begin{align*}
		c_{e_1}&=\sum_{i=1}^{d-1}\sum_{j=i+1}^d (j-i) x_i x_j,\\
		c_{e_2}&=\sum_{i=1}^{m-2}\sum_{j=i+1}^{m-1}(j-i)x_ix_j + \sum_{i=m}^{d-1}\sum_{j=i+1}^d (j-i)x_ix_j + \sum_{i=1}^{m-1}\sum_{j=m}^d (i+d-j)x_ix_j,\\
		c_{e_3}&=\sum_{i=1}^m\sum_{j=i+1}^{m+1} (j-i)x_ix_j + \sum_{i=m+2}^{d-1}\sum_{j=i+1}^d (j-i)x_ix_j + \sum_{i=1}^{m+1}\sum_{j=m+2}^d (i+d-j)x_ix_j,
	\end{align*}
	We observe that the claim is trivially true for $\frac{n}{3}\le \frac{d+6}{2}$, since $\Delta d(x)\ge\frac{d+6}{2}$ (tight for $n=6,d=6,x=(1,1,1,1,1,1)$). This means we can assume
	\begin{equation}\label{app:general:even:eq:n_bigger_3m}
		n>\frac{3(d+6)}{2}, \text{ ~~~~~or rather~~~~~ } n>3(m+3).
	\end{equation}
	
	Similar to the odd case, we perform a series of relaxations. For that, we define the node distributions $\x1,\x2,\x3,\x4$ such that for all $1\le i\le d$
	\begin{align*}
		\x1_i &\coloneqq \begin{cases}
			x_1 + \sum_{p=2}^{m-2}\frac{m-1-p}{m-2}(x_p-1)&\text{if }i=1,\\
			x_{m-1}+\sum_{p=2}^{m-2}\frac{p-1}{m-2}(x_p-1)&\text{if }i={m-1},\\
			x_m&\text{if }i=m,\\
			x_{m+1}&\text{if }i=m+1,\\
			x_{m+2}+\sum_{p=m+3}^{d-1}\frac{d-p}{m-2}(x_p-1)&\text{if }i=m+1,\\
			x_d + \sum_{p=m+3}^{d-1}\frac{p-m-2}{m-2}(x_p-1)&\text{if }i=d,\\
			1&\text{else,}
		\end{cases}\\
		&\qquad\text{(Without loss of generality, we assume } \x1_1\ge \x1_d\\
		&\qquad\text{and } (m-1)\x1_1+\x1_{m-1}=(m-1)\x1_d+\x1_{m+2}.)\\
		\x2_i &\coloneqq\begin{cases}
			\x1_1&\text{if }i=1\text{ or }i=d,\\
			\x1_{m-1}&\text{if }i=m-1\text{ or }i=m+2,\\
			\x1_m&\text{if }i=m\\
			\x1_{m+1}+(m-2)\left(\x1_1-\x1_d\right)&\text{if }i=m+1,\\
			1&\text{else},
		\end{cases}\\
		\x3_i &\coloneqq \begin{cases}
			\x2_1&\text{if }i=1\text{ or }i=d,\\
			\x2_{m-1}&\text{if }i=m-1\text{ or }i=m+2,\\
			\frac{\x2_m+\x2_{m+1}}{2}&\text{if }i=m\text{ or }i=m+1,\\
			1&\text{else},
		\end{cases}\\
		\x4_i &\coloneqq\begin{cases}
			\x3_1+\frac{1}{m-1}\left(\x3_{m-1}-1\right)&\text{if }i=1\text{ or }i=d,\\
			\x3_m+\frac{m-2}{m-1}\left(\x3_{m-1}-1\right)&\text{if }i=m\text{ or }i=m+1,\\
			1&\text{else}.
		\end{cases}
	\end{align*}
	Again, it is easy to see that $\Delta d(x)=\Delta d\left(\x1\right)=\dots=\Delta d\left(\x4\right)$. We show $c_{e_1}\left(\x{i}\right)\ge c_{e_2}\left(\x{i}\right)$ and $c_{e_1}\left(\x{i}\right)\ge c_{e_3}\left(\x{i}\right)$, for $1\le i\le 4$. Therefore, $\x4$ is a solution of the minimization problem.
	
	First, we observe that $c_{e_2}(x)-c_{e_1}(x)=c_{e_2}(\x1)-c_{e_1}(\x1)$ and $c_{e_3}(x)-c_{e_1}(x)=c_{e_3}(\x1)-c_{e_1}(\x1)$ follow the same way as in the odd case (with slight adjustments to the sum indices). This means that $\x1$ fulfills (\ref{app:general:stable:constr_even}).
	
	Next, we show that $\x2$ fulfills (\ref{app:general:stable:constr_even}), too. We have
	\begin{align*}
		c_{e_1}(\x1)={}&(m-2)\x1_1\x1_{m-1}+(m-1)\x1_1\x1_m+m\x1_1\x1_{m+1}+(m+1)\x1_1\x1_{m+2}\\*
		&+(2m-1)\x1_1\x1_d+\x1_{m-1}\x1_m+2\x1_{m-1}\x1_{m+1}+3\x1_{m-1}\x1_{m+2}\\*
		&+(m+1)\x1_{m-1}\x1_d+\x1_m\x1_{m+1}+2\x1_m\x1_{m+2}+m\x1_m\x1_d\\*
		&+\x1_{m+1}\x1_{m+2}+(m-1)\x1_{m+1}\x1_d+(m-2)\x1_{m+2}\x1_d\\*
		&+(m-3)(2m-1)\x1_1+(m-3)(m+1)\x1_{m-1}+(m-3)(m+1)\x1_m\\*
		&+(m-3)(m+1)\x1_{m+1}+(m-3)(m+1)\x1_{m+2}+(m-3)(2m-1)\x1_d\\*
		&+\frac{1}{3}(m-4)(m-3)(m-2)+(m-3)(m-3)(m+1),\\
		c_{e_2}(\x1)={}&(m-2)\x1_1\x1_{m-1}+(m+1)\x1_1\x1_m+m\x1_1\x1_{m+1}+(m-1)\x1_1\x1_{m+2}\\*
		&+\x1_1\x1_d+(2m-1)\x1_{m-1}\x1_m+(2m-2)\x1_{m-1}\x1_{m+1}+(2m-3)\x1_{m-1}\x1_{m+2}\\*
		&+(m-1)\x1_{m-1}\x1_d+\x1_m\x1_{m+1}+2\x1_m\x1_{m+2}+m\x1_m\x1_d\\*
		&+\x1_{m+1}\x1_{m+2}+(m-1)\x1_{m+1}\x1_d+(m-2)\x1_{m+2}\x1_d\\
		&+(m-3)(m-1)\x1_1+(m-3)(2m-3)\x1_{m-1}+(m-3)(2m+1)\x1_m\\
		&+(m-3)(2m-1)\x1_{m+1}+(m-3)(2m-3)\x1_{m+2}+(m-3)(m-1)\x1_d\\*
		&+\frac{1}{3}(m-4)(m-3)(m-2)+(m-3)(m-3)(m-1)\text{ and}\\
		c_{e_3}(\x1)={}&(m-2)\x1_1\x1_{m-1}+(m-1)\x1_1\x1_m+m\x1_1\x1_{m+1}+(m-1)\x1_1\x1_{m+2}\\*
		&+\x1_1\x1_d+\x1_{m-1}\x1_m+2\x1_{m-1}\x1_{m+1}+(2m-3)\x1_{m-1}\x1_{m+2}\\*
		&+(m-1)\x1_{m-1}\x1_d+\x1_m\x1_{m+1}+(2m-2)\x1_m\x1_{m+2}+m\x1_m\x1_d\\*
		&+(2m-1)\x1_{m+1}\x1_{m+2}+(m+1)\x1_{m+1}\x1_d+(m-2)\x1_{m+2}\x1_d\\*
		&+(m-3)(m-1)\x1_1+(m-3)(2m-3)\x1_{m-1}+(m-3)(2m-1)\x1_m\\*
		&+(m-3)(2m+1)\x1_{m+1}+(m-3)(2m-3)\x1_{m+2}+(m-3)(m-1)\x1_d\\*
		&+\frac{1}{3}(m-4)(m-3)(m-2)+(m-3)(m-3)(m-1).
	\end{align*}
	Similar to the odd case, we can further assume that $(m-1)\x1_1+\x1_{m-1}=(m-1)\x1_d+\x1_{m+2}$. If this is not the case and (without loss of generality) $(m-1)\x1_1+\x1_{m-1}>(m-1)\x1_d+\x1_{m+2}$, we could move weight from $\x1_{m+2}$ to $\x1_d$ without changing $\Delta d(\x1)$. This would increase $c_{e_1}(\x1)$ more than $c_{e_2}(\x1)$ and $c_{e_3}(\x1)$.
	
	Let (without loss of generality) $\x1_1\ge \x1_d$ and $y\coloneqq\x1_1-\x1_d$. Therefore, we have $\x1_d=\x1_1-y$ and $\x1_{m+2}=\x1_{m-1}+(m-1)y$. We see that
	\begin{align*}
		\forall 1\le i\le d\colon \x2_i=\begin{cases}
			\x1_m&\text{if } i=m,\\
			\x1_{m+1}+(m-2)y&\text{if } i=m+1,\\
			\x1_{m-1}&\text{if } i=m-1 \text{ or } i=m+2,\\
			\x1_1&\text{if } i=1 \text{ or } i=d,\\
			1&\text{else}.
		\end{cases}
	\end{align*}
	
	We now show that
	\begin{align*}
		\Delta_2 \x2 &\coloneqq  c_{e_1}(\x2)-c_{e_2}(\x2) \ge c_{e_1}(\x1)-c_{e_2}(\x1)=\Delta_2(\x1) \shortintertext{and}
		\Delta_3 \x2 &\coloneqq  c_{e_1}(\x2)-c_{e_3}(\x2) \ge c_{e_1}(\x1)-c_{e_3}(\x1)=\Delta_3(\x1).
	\end{align*} 
	
	We have
	\begin{align*}
		\Delta_2 \x1 ={}&c_{e_1}(\x1)-c_{e_2}(\x1)\\
		={}&-2\x1_1\x1_m+2\x1_1\x1_{m+2}+(2m-2)\x1_1\x1_d-(2m-2)\x1_{m-1}\x1_m\\*
		&-(2m-4)\x1_{m-1}\x1_{m+1}-(2m-6)\x1_{m-1}\x1_{m+2}+2\x1_{m-1}\x1_d\\*
		&+(m-3)m\x1_1-(m-3)(m-4)\x1_{m-1}-(m-3)m\x1_m\\*
		&-(m-3)(m-2)\x1_{m+1}-(m-3)(m-4)\x1_{m+2}+(m-3)m\x1_d\\*
		&+2(m-3)(m-3)\\
		={}&-2\x1_1\x1_m+2\x1_1\x1_{m-1}+2(m-1)\x1_1y+(2m-2)\x1_1\x1_1-(2m-2)\x1_1y\\*
		&-(2m-2)\x1_{m-1}\x1_m-(2m-4)\x1_{m-1}\x1_{m+1}-(2m-6)\x1_{m-1}\x1_{m-1}\\*
		&-(2m-6)\x1_{m-1}(m-1)y+2\x1_{m-1}\x1_1-2\x1_{m-1}y+(m-3)m\x1_1\\*
		&-(m-3)(m-4)\x1_{m-1}-(m-3)m\x1_m-(m-3)(m-2)\x1_{m+1}\\*
		&-(m-3)(m-4)\x1_{m-1}-(m-3)(m-4)(m-1)y+(m-3)m\x1_1\\*
		&-(m-3)my+2(m-3)(m-3)
		\shortintertext{and}
		\Delta_2 \x2 ={}&-2\x1_1\x1_m+2\x1_1\x1_{m-1}+(2m-2)\x1_1\x1_1-(2m-2)\x1_{m-1}\x1_m\\*
		&-(2m-4)\x1_{m-1}\x1_{m+1}-(2m-4)\x1_{m-1}(m-2)y-(2m-6)\x1_{m-1}\x1_{m-1}\\*
		&+2\x1_{m-1}\x1_1+(m-3)m\x1_1-(m-3)(m-4)\x1_{m-1}-(m-3)m\x1_m\\*
		&-(m-3)(m-2)\x1_{m+1}-(m-3)(m-2)(m-2)y\\*
		&-(m-3)(m-4)\x1_{m-1}+(m-3)m\x1_1+2(m-3)(m-3),
	\end{align*}
	and therefore
	\begin{align*}
		\Delta_2 \x2-\Delta_2 \x1={}&-(2m-4)\x1_{m-1}(m-2)y+(2m-6)\x1_{m-1}(m-1)y\\*
		&+2\x1_{m-1}y-(m-3)(m-2)(m-2)y\\*
		&+(m-3)(m-4)(m-1)y+(m-3)my\\
		={}&\x1_{m-1}y(-(2m-4)(m-2)+(2m-6)(m-1)+2)\\*
		&+(m-3)y(-(m-2)(m-2)+(m-4)(m-1)+m)\\
		={}&\x1_{m-1}y\cdot 0 + (m-3)y\cdot 0\\
		={}&0.
	\end{align*}
	
	Furthermore, we see that
	\begin{align*}
		\Delta_3 \x1 ={}&c_{e_1}(\x1)-c_{e_3}(\x1)\\
		={}&2\x1_1\x1_{m+2}+(2m-2)\x1_1\x1_d-(2m-6)\x1_{m-1}\x1_{m+2}+2\x1_{m-1}\x1_d\\*
		&-(2m-4)\x1_m\x1_{m+2}-(2m-2)\x1_{m+1}\x1_{m+2}-2\x1_{m+1}\x1_d\\*
		&+(m-3)m\x1_1-(m-3)(m-4)\x1_{m-1}-(m-3)(m-2)\x1_m\\*
		&-(m-3)m\x1_{m+1}-(m-3)(m-4)\x1_{m+2}+(m-3)m\x1_d\\*
		&+2(m-3)(m-3)\\
		={}&2\x1_1\x1_{m-1}+2\x1_1(m-1)y+(2m-2)\x1_1\x1_1-(2m-2)\x1_1y\\*
		&-(2m-6)\x1_{m-1}\x1_{m-1}-(2m-6)\x1_{m-1}(m-1)y+2\x1_{m-1}\x1_1-2\x1_{m-1}y\\*
		&-(2m-4)\x1_m\x1_{m-1}-(2m-4)\x1_m(m-1)y-(2m-2)\x1_{m+1}\x1_{m-1}\\*
		&-(2m-2)\x1_{m+1}(m-1)y-2\x1_{m+1}\x1_1+2\x1_{m+1}y\\*
		&+(m-3)m\x1_1-(m-3)(m-4)\x1_{m-1}-(m-3)(m-2)\x1_m\\*
		&-(m-3)m\x1_{m+1}-(m-3)(m-4)\x1_{m-1}-(m-3)(m-4)(m-1)y\\*
		&+(m-3)m\x1_1-(m-3)my+2(m-3)(m-3)
		\shortintertext{and}
		\Delta_3 \x2={}&2\x1_1\x1_{m-1}+(2m-2)\x1_1\x1_1-(2m-6)\x1_{m-1}\x1_{m-1}+2\x1_{m-1}\x1_1\\*
		&-(2m-4)\x1_m\x1_{m-1}-(2m-2)\x1_{m+1}\x1_{m-1}-(2m-2)(m-2)y\x1_{m-1}\\*
		&-2\x1_{m+1}\x1_1-2(m-2)y\x1_1+(m-3)m\x1_1-(m-3)(m-4)\x1_{m-1}\\*
		&-(m-3)(m-2)\x1_m-(m-3)m\x1_{m+1}-(m-3)m(m-2)y\\*
		&-(m-3)(m-4)\x1_{m-1}+(m-3)m\x1_1+2(m-3)(m-3),
	\end{align*}
	and therefore
	\begin{align*}
		\Delta_3 \x2-\Delta_3 \x1={}&(2m-6)\x1_{m-1}(m-1)y+2\x1_{m-1}y+(2m-4)\x1_m(m-1)y\\*
		&-(2m-2)(m-2)y\x1_{m-1}+(2m-2)\x1_{m+1}(m-1)y-2(m-2)y\x1_1\\*
		&-2\x1_{m+1}y-(m-3)m(m-2)y+(m-3)(m-4)(m-1)y\\*
		&+(m-3)my\\
		={}&-\x1_1y(2m-4)\\*
		&+\x1_{m-1}y((2m-6)(m-1)+2-(2m-2)(m-2))\\*
		&+\x1_my(2m-4)(m-1)\\*
		&+\x1_{m+1}y((2m-2)(m-1)-2)\\*
		&+(m-3)y(-(m-2)m+(m-4)(m-1)+m)\\
		={}&2(m-2)y(-\x1_1-\x1_{m-1}+(m-1)\x1_m+m\x1_{m+1}-(m-3))\\
		\ge{}&0.
	\end{align*}
	The inequality in the last step comes from $m\ge 3$, $y\ge 0$, $\x1_1+\x1_{m-1}\le\frac{n}{3}$ and $\x1_m+\x1_{m+1}\ge\frac{n}{3}$.
	
	Since $\x1$ fulfills (\ref{app:general:stable:constr_even}) and we have $\Delta_2 \x2 - \Delta_2 \x1=0$ and $\Delta_3 \x2 -\Delta_3 \x1\ge 0$, the node distribution $\x2$ fulfills (\ref{app:general:stable:constr_even}), too.
	
	Now, we show that $\x3$ fulfills (\ref{app:general:stable:constr_even}), as well. Without loss of generality, we can assume $\x2_{m+1}\ge \x2_m$.
	
	We have
	\begin{align*}
		\Delta_3 \x2={}&4\x2_1\x2_{m-1}+(2m-2)\x2_1\x2_1-(2m-6)\x2_{m-1}\x2_{m-1}-(2m-4)\x2_m\x2_{m-1}\\*
		&-(2m-2)\x2_{m+1}\x2_{m-1}-2\x2_{m+1}\x2_1+2(m-3)m\x2_1-2(m-3)(m-4)\x2_{m-1}\\*
		&-(m-3)(m-2)\x2_m-(m-3)m\x2_{m+1}+2(m-3)(m-3)
		\intertext{and}
		\Delta_3 \x3={}&4\x2_1\x2_{m-1}+(2m-2)\x2_1\x2_1-(2m-6)\x2_{m-1}\x2_{m-1}\\*
		&-(2m-4)\x2_{m-1}\frac{\x2_m+\x2_{m+1}}{2}-(2m-2)\x2_{m-1}\frac{\x2_m+\x2_{m+1}}{2}-2\x2_1\frac{\x2_m+\x2_{m+1}}{2}\\*
		&+2(m-3)m\x2_1-2(m-3)(m-4)\x2_{m-1}-(m-3)(m-2)\frac{\x2_m+\x2_{m+1}}{2}\\*
		&-(m-3)m\frac{\x2_m+\x2_{m+1}}{2}+2(m-3)(m-3)\\
		={}&4\x2_1\x2_{m-1}+(2m-2)\x2_1\x2_1-(2m-6)\x2_{m-1}\x2_{m-1}-(2m-3)\x2_{m-1}\x2_m\\*
		&-(2m-3)\x2_{m-1}\x2_{m+1}-\x2_1\x2_m-\x2_1\x2_{m+1}+2(m-3)m\x2_1\\*
		&-2(m-3)(m-4)\x2_{m-1}-(m-3)(m-1)\x2_m-(m-3)(m-1)\x2_{m+1}\\*
		&+2(m-3)(m-3),
	\end{align*}
	and therefore
	\begin{align*}
		\Delta_3 \x3-\Delta_3 \x2={}&-\x2_{m-1}\x2_m+\x2_{m-1}\x2_{m+1}-\x2_1\x2_m+\x2_1\x2_{m+1}\\*
		&-(m-3)\x2_m+(m-3)\x2_{m+1}\\
		={}&(\x2_{m+1}-\x2_m)(\x2_{m-1}+\x2_1+m-3)\\
		\ge{}&0.
	\end{align*}
	We also see that $\Delta_2 \x3-\Delta_2 \x3\ge 0$ because $\x3$ is symmetric. Since $\x2$ fulfills (\ref{app:general:stable:constr_even}), $\x3$ does, too. This also means that the two inequalities from (\ref{app:general:stable:constr_even}) are equivalent.
	
	Finally, we show that $\x4$ fulfills (\ref{app:general:stable:constr_even}). Let $y\coloneqq \frac{\x3_{m-1}-1}{m-1}$. We see that
	\begin{align*}
		\forall 1\le i\le d\colon \x4_i=\begin{cases}
			\x3_{m}+(m-2)y&\text{if } i=m\text{ or } i=m+1,\\
			\x3_1+y&\text{if } i=1 \text{ or } i=d,\\
			1&\text{else}.
		\end{cases}
	\end{align*}
	
	We have
	\begin{align*}
		\Delta_3 \x3 ={}&4\x3_1\x3_{m-1}+(2m-2)\x3_1\x3_1-(2m-6)\x3_{m-1}\x3_{m-1}-(4m-6)\x3_{m-1}\x3_m\\*
		&-2\x3_1\x3_m+2(m-3)m\x3_1-2(m-3)(m-4)\x3_{m-1}-(m-3)(2m-2)\x3_m\\*
		&+2(m-3)(m-3)\\
		={}&4\x3_1+4\x3_1(m-1)y+(2m-2)\x3_1\x3_1-(2m-6)(1+(m-1)y)^2\\*
		&-(4m-6)\x3_m-(4m-6)\x3_m(m-1)y-2\x3_1\x3_m+2(m-3)m\x3_1\\*
		&-2(m-3)(m-4)-2(m-3)(m-4)(m-1)y-(m-3)(2m-2)\x3_m\\*
		&+2(m-3)(m-3)\\
		={}&4\x3_1+(2m-2)\x3_1\x3_1-(2m-6)-(4m-6)\x3_m-2\x3_1\x3_m+2(m-3)m\x3_1\\*
		&-2(m-3)(m-4)-(m-3)(2m-2)\x3_m+2(m-3)(m-3)\\*
		&+y(m-1)(4\x3_1-2(2m-6)-(4m-6)\x3_m-2(m-3)(m-4)))\\*
		&-y^2(m-1)^2(2m-6)
		\shortintertext{and}
		\Delta_3 \x4={}&4\x3_1+4y+(2m-2)(\x3_1+y)^2-(2m-6)-(4m-6)\x3_m\\*
		&-(4m-6)(m-2)y-2(\x3_1+y)(\x3_m+(m-2)y)+2(m-3)m\x3_1\\*
		&+2(m-3)my-2(m-3)(m-4)-(m-3)(2m-2)\x3_m\\*
		&-(m-3)(2m-2)(m-2)y+2(m-3)(m-3)\\
		={}&4\x3_1+(2m-2)\x3_1\x3_1-(2m-6)-(4m-6)\x3_m-2\x3_1\x3_m+2(m-3)m\x3_1\\*
		&-2(m-3)(m-4)-(m-3)(2m-2)\x3_m+2(m-3)(m-3)\\*
		&+y(4+2(2m-2)\x3_1-(4m-6)(m-2)-2\x3_1(m-2)-2\x3_m\\*
		&\qquad+2(m-3)m-(m-3)(2m-2)(m-2))\\*
		&+y^2((2m-2)-2(m-2))\\
		={}&4\x3_1+(2m-2)\x3_1\x3_1-(2m-6)-(4m-6)\x3_m-2\x3_1\x3_m+2(m-3)m\x3_1\\*
		&-2(m-3)(m-4)-(m-3)(2m-2)\x3_m+2(m-3)(m-3)\\*
		&+y(2m\x3_1-2\x3_m+4-(4m-6)(m-2)+2(m-3)m\\*
		&\qquad-(m-3)(2m-2)(m-2))\\*
		&+2y^2,
	\end{align*}
	and therefore
	\begin{align*}
		\Delta_3 \x4-\Delta_3 \x3={}&y(-(2m-4)\x3_1+(2m-4)(2m-1)\x3_m+4-(4m-6)(m-2)\\*
		&\qquad+2(m-3)m-(m-3)(2m-2)(m-2)\\*
		&\qquad+2(2m-6)(m-1)+2(m-3)(m-4)(m-1))\\*
		&+y^2(2+(m-1)^2(2m-6))\\
		={}&y(2m-4)(-\x3_1+(2m-1)\x3_m-(m-2))\\*
		&+y^2(2+(m-1)^2(2m-6))\\
		\ge{}&0.
	\end{align*}
	The last step follows from $\x3_m\ge \x3_1$, $m\ge 3$ and $y\ge 0$.
	
	Since $\x3$ and $\x4$ are symmetric, we also have $\Delta_2(\x4)-\Delta_2(\x3)\ge 0$, and with $\x3$ fulfilling (\ref{app:general:stable:constr_even}), the node distribution $\x4$ fulfills (\ref{app:general:stable:constr_even}), too. Because of $\x4$ being a node distribution with only two variables ($\x4_1=\x4_d$ and $\x4_m=\x4_{m+1}$), we can simplify $\Delta \x4$ and $\Delta d(\x4)$ to
	\begin{align*}
		\Delta \x4 \coloneqq {}& \Delta_2 \x4 = \Delta_3 \x4\\
		={}&4\x4_1+(2m-2)\x4_1\x4_1-(2m-6)-(4m-6)\x4_m-2\x4_1\x4_m\\*
		&+2(m-3)m\x4_1-2(m-3)(m-4)-(m-3)(2m-2)\x4_m\\*
		&+2(m-3)(m-3)\\
		={}&2(m-1)\x4_1\x4_1+2(m-1)(m-2)\x4_1-2\x4_1\x4_m-2m(m-2)\x4_m
		\intertext{and}
		\Delta d(\x4)={}&2(m-1)\x4_1+(m-2)(m-1)+2.
	\end{align*}
	
	We now prove the lower bound.
	
	We see that $\Delta d$ is now only dependent on $\x4_1$ which means we have to minimize $\x4_1$. Since $\x4$ sums up to $n$, we have $n=2\x4_1+2\x4_m+2m-4$, and therefore $\x4_m=\frac{n}{2}-\x4_1-m+2$. Substituting this into $\Delta \x4$ yields
	\begin{align*}
		\Delta \x4={}&2(m-1)\x4_1\x4_1 + 2(m-1)(m-2)\x4_1 - 2\x4_1\left(\frac{n}{2}-\x4_1-m+2\right)\\*
		&-2m(m-2)\left(\frac{n}{2}-\x4_1-m+2\right)\\
		={}&2m\x4_1\x4_1+(4m(m-2)-n)\x4_1-m(m-2)(n-2m+4).
	\end{align*}
	We know that $\x4$ fulfills (\ref{app:general:stable:constr_even}) if and only if $\Delta \x4\ge 0$. Solving the quadratic inequality leaves us with
	\begin{align*}
		\x4_1\ge \frac{n}{4m}-(m-2)+\sqrt{\left(\frac{n}{4m}-(m-2)\right)^2+(m-2)\left(\frac{n}{2}-(m-2)\right)}.
	\end{align*}
	Because of (\ref{app:general:even:eq:n_bigger_3m}), we see that the term under the root is larger than
	\begin{align*}
		&0+(m-2)\left(\frac{3(m+3)}{2}-m+2\right)\\
		\ge{}&\frac{1}{2}(m-2)(m-2).
	\end{align*}
	We therefore obtain
	\begin{align*}
		\x4_1&\ge\frac{n}{4m}-(m-2)+\frac{\sqrt{2}}{2}(m-2)\\
		&\ge\frac{n}{4m}-\frac{m-2}{2},
	\end{align*}
	and finally
	\begin{align*}
		\Delta d(\x4)\ge 2(m-1)\left(\frac{n}{4m}-\frac{m-2}{2}\right)+(m-2)(m-1)+2\ge\frac{m-1}{m}\cdot\frac{n}{2}\ge\frac{n}{3},
	\end{align*}
	which concludes the proof.
\end{proof}
Next, we show that we can find a SMRCST in polynomial time via \Cref{algo} and even guarantee some bounds on the social welfare of the resulting network.
\begin{algorithm}
	\SetAlgoLined
	\SetKwData{T}{T}\SetKwData{H}{H}\SetKwData{P}{P}\SetKwFunction{greedy}{greedyLongPath}
	\KwIn{connected host network \H}
	\KwOut{SMRCST \T}
	\P$\leftarrow$ \greedy(\H)\;
	\T$\leftarrow$\emph{ extend \P to form a spanning tree of \H}\;\label{algo:extend}
	\While{$\exists e\in E_{\T},e'\in E_{\H}\setminus E_{\T}:d_{\T-e+e'}(V_{\H},V_{\H})>d_{\T}(V_{\H},V_{\H})$}{
		\T$\leftarrow\T-e+e'$
	}
	\caption{Computes a SMRCST for a given connected host network.}\label{algo}
\end{algorithm}
Our algorithm employs a greedy algorithm developed by Karger et al.~\cite{karger1997} which can find a path of length at least $\frac{|E_H|}{|V_H|}$ in $\mathcal{O}(|E_H|)$ as a subroutine for initialization. We call this subroutine \emph{greedyLongPath}. This will help us to derive bounds on the total distances of the computed SMRCST later. For extending the path to a spanning tree in line \ref{algo:extend} of \Cref{algo}, we can simply iterate over all edges and add them to the network if they do not close a cycle.
\begin{restatable}{thm}{algo}
	Let $H$ be a connected network containing $n$ nodes and $m$ edges. Then \Cref{algo} finds a Swap-Maximal Routing-Cost Spanning Tree of $H$ in runtime $\mathcal{O}(n^5m)$.
\end{restatable}
\begin{proof}
	It is easy to see that, by construction, $T$ is always a spanning tree of $H$. The condition in the while-loop ensures that all possible swaps are tried. This means that the while-loop ends if and only if $T$ is a SMRCST. Therefore if the while-loop stops, the result is correct.
	
	In every iteration in which the while-loop does not stop, the total distances of $T$ increase by at least 1. Since the tree maximizing the total distances is the path, we get its total distances $\frac{1}{3}(n-1)n(n+1)\in\mathcal{O}(n^3)$ as an upper bound for the number of iterations.
	
	The runtime of \Cref{algo} is clearly dominated by the while-loop. Since $T$ has $n-1$ edges which can be removed and $E_H\setminus E_T$ has $\mathcal{O}(m)$ possible edges to add, the number of possible swaps is in $\mathcal{O}(nm)$. For each swap, the total distances can be computed in $\mathcal{O}(n)$ \cite{mohar1988}. Therefore computing the condition can be done in $\mathcal{O}(n^2m)$. Altering the current solution in the body of the while-loop only takes $\mathcal{O}(n)$ when using adjacency lists. Since there are at most $\mathcal{O}(n^3)$ iterations of the while-loop, the overall runtime is in $\mathcal{O}(n^5m)$.
\end{proof}
\noindent For the $K$-SDNCG, the social optima were also stable. For general host networks, this is not necessarily the case. However, we can show that for $\alpha\le\frac{n}{3}$ there are stable states which approximate the social welfare better than with the trivial factor of $\mathcal{O}(n)$.
\begin{restatable}[OPT-Approximation via the MRCST]{thm}{generalOptApprox}\label{thm:gen:approx}
	Let $H$ be a connected host network containing $n$ nodes and $m$ edges and $T$ be the MRCST of $H$.
	\begin{bracketenumerate}
		\item We have $\frac{\SW(\OPT_H)}{\SW(T)}\in\mathcal{O}\left(\frac{m}{n}\right)$.\label{thm:gen:approx:sparse}
		\item For $\alpha\in\mathcal{O}\big(\frac{n^2}{m}\big)$, we have $\frac{\SW(\OPT_H)}{\SW(T)}\in\mathcal{O}(1)$.\label{thm:gen:approx:constant}
		\item For $\alpha\in\omega\big(\frac{n^2}{m}\big)$, we have $\frac{\SW(\OPT_H)}{\SW(T)}\in\mathcal{O}\big(\min\big\{\frac{m}{n},\alpha\frac{n}{m}\big\}+1\big)$.\label{thm:gen:approx:dense}
	\end{bracketenumerate}

\end{restatable}
\begin{proof}[Proof of (\ref{thm:gen:approx:sparse})]
	Let $V\coloneqq V_H$ and $O\coloneqq \OPT_H$. Since $T$ is a MRCST of $H$, we have $d_T(V,V)\ge d_O(V,V)$. This yields
	\begin{equation*}
		\frac{\SW(O)}{\SW(T)}=\frac{2\alpha|E_O|+d_O(V,V)}{2\alpha|E_T|+d_T(V,V)}\le\frac{2\alpha m+d_T(V,V)}{2\alpha(n-1)+d_T(V,V)}\le\frac{m}{n-1}+1\in\mathcal{O}\left(\frac{m}{n}\right).\qedhere
	\end{equation*}
\end{proof}
\begin{proof}[Proof of (\ref{thm:gen:approx:constant})]
	Let $V\coloneqq V_H$ and $O\coloneqq \OPT_H$. We have $2\alpha|E_T|\le 2\alpha|E_O|\le2\alpha m\in\mathcal{O}(n^2)$ and $d_O(V,V)\in\Omega(n^2)$. We also know that $d_O(V,V)\le d_T(V,V)$ because $T$ is a MRCST. This yields
	\begin{equation*}
		\frac{\SW(O)}{\SW(T)}=\frac{2\alpha|E_O|+d_O(V,V)}{2\alpha|E_T|+d_T(V,V)}\in\mathcal{O}\left(\frac{d_O(V,V)}{d_T(V,V)}\right)=\mathcal{O}(1).\qedhere
	\end{equation*}
\end{proof}
\begin{proof}[Proof of (\ref{thm:gen:approx:dense})]
	Let $V\coloneqq V_H$ and $O\coloneqq \OPT_H$.  \Cref{algo} uses the greedyLongPath algorithm to construct an initial tree that contains a path of at least length $l\ge\frac{m}{n}$ \cite{karger1997}. Every node in that tree has a distance of at least $\frac{l}{3}$ to at least $\frac{l}{3}$ of the nodes on the long path. This means that the total distances are at least $n\left(\frac{l}{3}\right)^2\in\Omega\left(\frac{m^2}{n}\right)$. Since \Cref{algo} only increases the total distances, this is a lower bound for every solution found by the algorithm, and therefore also for the MRCST. With this, we get
	\begin{align*}
		\frac{\SW(O)}{\SW(T)}&=\frac{2\alpha|E_O|+d_O(V,V)}{2\alpha|E_T|+d_T(V,V)}\le\frac{2\alpha m}{2\alpha(n-1)+d_T(V,V)}+\frac{d_T(V,V)}{2\alpha(n-1)+d_T(V,V)}\\
		&\in\mathcal{O}\left(\frac{\alpha m}{\alpha n+\frac{m^2}{n}}+1\right)=\mathcal{O}\left(\min\left\{\frac{m}{n},\alpha\frac{n}{m}\right\}+1\right).\qedhere
	\end{align*}
\end{proof}
\noindent Since finding the MRCST is NP-hard \cite{CAMERINI83}, these are only existence results. However, the next theorem yields a bound for dense networks and the computed SMRCST from \Cref{algo}.
\begin{restatable}{thm}{generalSmaxrcstApprox}\label{general:smaxrcst:approx}
	Let $H$ be a connected host network containing $n$ nodes and $m$ edges and $T$ be the SMRCST obtained by \Cref{algo}. Then for $\alpha\in\mathcal{O}(n)$, we have $\frac{\SW(\OPT_H)}{\SW(T)}\in\mathcal{O}\big(\frac{n^4}{m^2}\big)$.
\end{restatable}
\begin{proof}
Let $V\coloneqq V_H$ and $O\coloneqq \OPT_H$. We can trivially upper bound the distances and number of edges with $d_O(V,V)\in\mathcal{O}(n^3)$ and $|E_O|\le m$. Since $T$ is a result of \Cref{algo}, we have $d_T(V,V)\in\Omega\left(\frac{m^2}{n}\right)$, as seen in the previous proof. This yields
\begin{align*}
	\frac{\SW(O)}{\SW(T)}&=\frac{2\alpha|E_O|+d_O(V,V)}{2\alpha|E_T|+d_T(V,V)}\le\frac{2\alpha m}{2\alpha(n-1)+\frac{m^2}{n}}+\frac{n^3}{2\alpha(n-1)+\frac{m^2}{n}}\\
	&\in\mathcal{O}\left(\frac{n^2}{m}+\frac{n^4}{m^2}\right)=\mathcal{O}\left(\frac{n^4}{m^2}\right).\qedhere
\end{align*}
\end{proof}
\noindent This means, for $\alpha\le\frac{n}{3}$ and a dense host network, we can compute a state which is pairwise stable and also has a favorable social welfare.
\subsection{Price of Anarchy and Price of Stability}
\label{sec:general:poaPos}
We derive several bounds on the PoA and the PoS for the SDNCG. For the $K$-SDNCG, the PoA is already quite high for small $\alpha$. The next Theorem shows that this gets even worse for general host networks since the PoA is linear up to $\alpha\le n$ and super-constant for $\alpha\in o(n^2)$.
\begin{restatable}[Price of Anarchy]{thm}{generalPoa}\label{thm:gen:poa}\,
	\begin{bracketenumerate}
		\item The Price of Anarchy is in $\mathcal{O}(n)$.\label{thm:gen:poa:upper_bound}
		\item For $\alpha<1$, the Price of Anarchy is in $\Theta(n)$.\label{thm:gen:poa:alpha_smaller_one}
		\item For $1\le\alpha\le n$, the Price of Anarchy is in $\Theta(n)$.\label{thm:gen:poa:alpha_smaller_n}
		\item For $n<\alpha\le n^2$, the Price of Anarchy is in $\Omega\left(\frac{n^2}{\alpha}\right)$.\label{thm:gen:poa:alpha_smaller_n2}
		\item For $\frac{1}{4}(n-1)^2<\alpha\le \frac{1}{24}(n-2)n(n+2)$, the Price of Anarchy is in $\Theta(1)$.\label{thm:gen:poa:alpha_smaller_n3}
		\item For $\alpha>\frac{1}{24}(n-2)n(n+2)$, the Price of Anarchy is 1.\label{thm:gen:poa:large_alpha}
	\end{bracketenumerate}

\end{restatable}
\begin{proof}[Proof of (\ref{thm:gen:poa:upper_bound}) and (\ref{thm:gen:poa:alpha_smaller_one})]
	This follows in the same way as (\ref{thm:com:poa:upper_bound}) and (\ref{thm:com:poa:alpha_smaller_one}) of \Cref{thm:com:poa}.
\end{proof}
\begin{proof}[Proof of (\ref{thm:gen:poa:alpha_smaller_n}) and (\ref{thm:gen:poa:alpha_smaller_n2})]
	\begin{figure}
		\centering
		\includegraphics[width=0.3\textwidth]{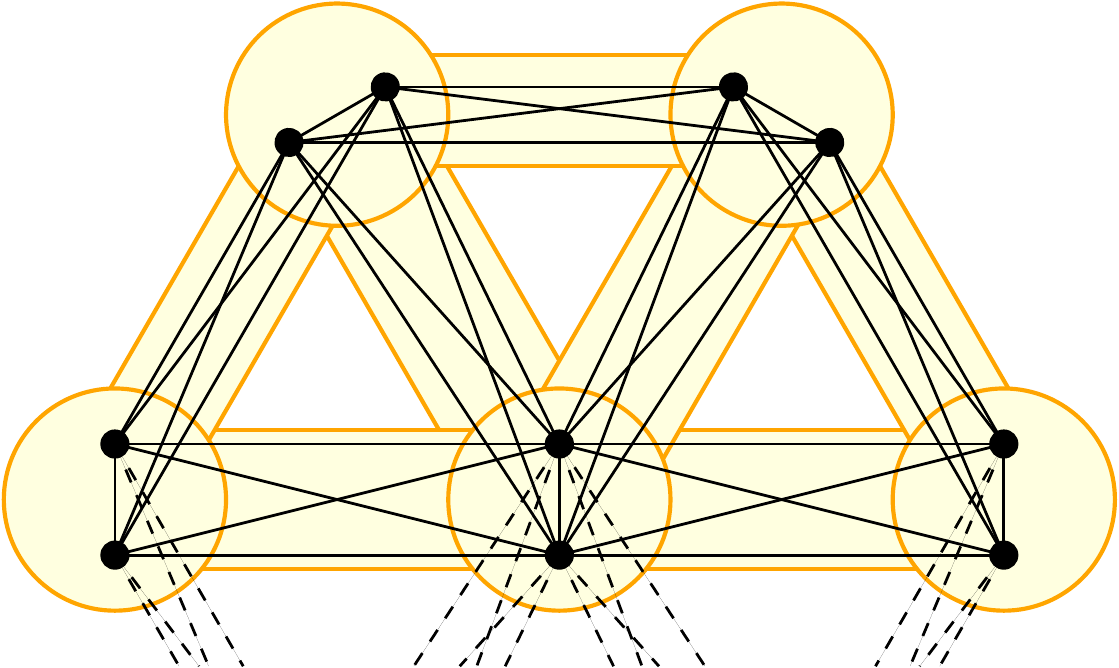}
		\caption{This figure shows a clique network (black) for a wheel network (yellow).}
		\label{fig:wheel}
	\end{figure}
	Let $W=(V_W,E_W)$ be a wheel network on $n'\coloneqq \left\lfloor\frac{n}{2}\right\rfloor$ nodes, i.e.,
	\begin{align*}
		V_W&\coloneqq \{v_1,\dots,v_{n'}\}\quad\text{and}\\
		E_W&\coloneqq \{\{v_1,v_i\}\mid 2\le i\le n'\}\cup\{\{v_i,v_{i+1}\}\mid 2\le i\le n'\}\cup\{\{v_2,v_{n'}\}\}.
	\end{align*}
	We then define the host network $H$ as the clique network obtained by replacing every node of $W$ by a clique of size 2. (See \Cref{fig:wheel} for an illustration.) For odd $n$, we instead replace the central node by a clique of size 3. We see that $H$ contains $n$ nodes, $\Theta(n)$ edges, and most importantly a Hamilton path. We also know that $H$ is stable because of \Cref{clique_network} and since no edge can be added. This yields the following lower bound for the Price of Anarchy
	\begin{align*}
		\POA\ge\frac{SW(P)}{SW(H)}=\frac{\alpha(n-1)+\Theta(n^3)}{\alpha\Theta(n)+\Theta(n^2)}\in\Omega\left(\min\left\{n,\frac{n^2}{\alpha}\right\}\right),
	\end{align*}
	which proves the claim.
\end{proof}
\begin{proof}[Proof of (\ref{thm:gen:poa:alpha_smaller_n3})]
	Let $H$ be a (connected) host network and $O\coloneqq \OPT_H$ and $V\coloneqq V_H$. From \Cref{thm:gen:stable:large_alpha} we know that $H$ itself is the only pairwise stable network. Since we can bound $|E_O|\le |E_H|$ and $d_O(V,V)\le n^3$, we obtain
	\begin{align*}
		\POA_n\le\frac{\SW(O)}{\SW(H)}=\frac{2\alpha|E_O|+d_O(V,V)}{2\alpha|E_H|+d_H(V,V)}\le\frac{2\alpha|E_H|+n^3}{2\alpha|E_H|}\in\Theta\left(1+\frac{n^3}{\alpha n}\right)=\Theta(1).
	\end{align*}
\end{proof}
\begin{proof}[Proof of (\ref{thm:gen:poa:large_alpha})]
	This follows directly from the host network being socially optimal and the only stable network (see \Cref{thm:gen:opt} and \Cref{thm:gen:stable}).
\end{proof}

\begin{restatable}[Price of Stability]{thm}{generalPos}\label{thm:gen:pos}\,
	\begin{bracketenumerate}
		\item The Price of Stability is in $\mathcal{O}(n)$.\label{thm:gen:pos:upper_bound}
		\item For $\alpha\le 1$, the Price of Stability is 1.\label{thm:gen:pos:alpha_smaller_one}
		\item For $1<\alpha\le\frac{n}{3}$, the Price of Stability is in $\mathcal{O}(\sqrt{n})$.\label{thm:gen:pos:alpha_smaller_n_3}
		\item For $\frac{1}{4}(n-1)^2<\alpha\le \frac{1}{24}(n-2)n(n+2)$, the Price of Stability is in $\Theta(1)$.\label{thm:gen:pos:alpha_smaller_n3}
		\item For $\alpha>\frac{1}{24}(n-2)n(n+2)$, the Price of Stability is 1.\label{thm:gen:pos:large_alpha}
	\end{bracketenumerate}

\end{restatable}
\begin{proof}[Proof of (\ref{thm:gen:pos:upper_bound}), (\ref{thm:gen:pos:alpha_smaller_n3}), and (\ref{thm:gen:pos:large_alpha})]
	This follows from \Cref{thm:gen:poa} and the Price of Anarchy being an upper bound for the Price of Stability.
\end{proof}
\begin{proof}[Proof of (\ref{thm:gen:pos:alpha_smaller_one})]
	This follows directly from the MRCST being socially optimal and stable (see \Cref{thm:gen:opt} and \Cref{thm:gen:stable}).
\end{proof}
\begin{proof}[Proof of (\ref{thm:gen:pos:alpha_smaller_n_3})]
	Let $H$ be a (connected) host network and $T$ be the MRCST of $H$. Let furthermore $O\coloneqq \OPT_H$ and $V\coloneqq V_H$. From \Cref{general:stable:huge_proof} we know that $T$ is stable. From \Cref{thm:gen:opt} and \Cref{thm:gen:approx}, we know that
	\begin{equation*}
		\frac{\SW(O)}{\SW(T)}\in\mathcal{O}\left(\min\left\{\frac{m}{n},\alpha\frac{n}{m}\right\}\right)\le\mathcal{O}\left(\min\left\{\frac{m}{n},\frac{n^2}{m}\right\}\right)\le\mathcal{O}\left(\sqrt{n}\right).\qedhere
	\end{equation*}
\end{proof}

\section{Conclusion}
\label{sec:conclusion}
We introduced and analyzed a natural game-theoretic model for network formation governed by social distancing. Besides modeling this timely issue, our model resembles the inverse compared to the well-known (bilateral) Network Creation Game~\cite{fabrikant2003,CP05}. Thus, via our analysis we could explore the impact of inverting the utility function in a non-trivial strategic game. We find that this inverts some of the properties, like the rough structure of optimum states, while it also yields non-obvious insights. First of all, for the variant with non-complete host networks we could show a strong equilibrium existence result, whereas no such result is known for the inverse model. Moreover, we established that the PoA is significantly higher in the ($K$-)SDCNG compared to the (bilateral) NCG. This demonstrates that the impact of the agents' selfishness is higher under social distancing, which calls for external coordination.    

The most obvious open question for future work is to settle the equilibrium existence. Do pairwise stable states exist for all connected host networks $H$ and $\alpha$? Another research direction would be to consider the unilateral variant of the SDNCG. While this no longer realistically models the formation of social networks, it might still yield interesting insights and it allows for studying stronger solution concepts like the Nash equilibrium or strong Nash equilibria, similar to~\cite{janus2017,andelman2009}. Also, altering the utility function, e.g., to using the maximum distance instead of the summed distances, or the probability of infection, similar to \cite{blume2011}, seems promising. Finally, also considering weighted host networks, as in~\cite{bilo2019}, where the edge weight models the benefit of the social interaction, would be an interesting generalization.



\bibliography{social-distancing-network-creation}


\end{document}